%% file: main.tex
\documentclass[11pt]{article}
\usepackage[margin=1in]{geometry}
\usepackage{libertine}

\usepackage{stmaryrd}
\usepackage{datetime}
\usepackage{multirow}
\usepackage{booktabs}
\usepackage[shortlabels]{enumitem}
\usepackage{wrapfig}
\usepackage{subcaption}
\usepackage[T1]{fontenc}
\usepackage[utf8]{inputenc}
\usepackage{pgfplots}
\usepackage{amsmath, amssymb}
\usepackage{amsthm}
\usepackage{color}
\usepackage{cases}
\usepackage{xparse}
\usepackage{xargs}
\usepackage{appendix}
\usepackage[boxed,commentsnumbered,noend]{algorithm2e}
\usepackage{algpseudocode}
\usepackage{verbatim, xspace}
\usepackage{tikz}
\usepackage{mathtools}

\usepackage{hyperref}
\usepackage[capitalize]{cleveref}
\usepackage{pgfplots}

\usepackage{thmtools}
\usepackage{thm-restate}

\usepackage{xcolor}

\newcommand{\citet}[1]{\cite{#1}}
\SetAlgorithmName{Isolation scheme}{isolation scheme}{List of isolaiton schemes}
\usepackage[disable,colorinlistoftodos,prependcaption,textsize=tiny]{todonotes}
\newcommandx{\unsure}[2][1=]{\todo[linecolor=green,backgroundcolor=green!25,bordercolor=green,#1]{\normalsize #2}}
\newcommandx{\improvement}[2][1=]{\todo[inline,linecolor=blue,backgroundcolor=blue!05,bordercolor=blue,#1]{\normalsize #2}}
\newcommandx{\info}[2][1=]{\todo[linecolor=yellow,backgroundcolor=yellow!25,bordercolor=yellow,#1]{#2}}
\newcommandx{\floatmodel}[2][1=]{\todo[inline,linecolor=red,backgroundcolor=yellow!25,bordercolor=yellow,#1]{#2}}
\newcommandx{\thiswillnotshow}[2][1=]{\todo[disable,#1]{#2}}
\newcommandx{\karol}[2][1=]{\todo[inline,linecolor=blue,backgroundcolor=blue!25,bordercolor=blue,caption={\normalsize \textbf{Karol}},#1]{\normalsize #2}}
\newcommandx{\jana}[2][1=]{\todo[inline,linecolor=red,backgroundcolor=red!25,bordercolor=red,caption={\normalsize
\textbf{jana}},#1]{\normalsize #2}}
\newcommandx{\michal}[2][1=]{\todo[inline,linecolor=gray,backgroundcolor=red!25,bordercolor=red,caption={\normalsize \textbf{Micha\l{}}},#1]{\normalsize #2}}

\newtheorem{theorem}{Theorem}
\newtheorem{definition}[theorem]{Definition}
\newtheorem{lemma}[theorem]{Lemma}

\newtheorem{claim}[theorem]{Claim}
\newtheorem{proposition}[theorem]{Proposition}

\newtheorem*{maingoal*}{Open Question}

\numberwithin{theorem}{section}
\numberwithin{lemma}{section}
\numberwithin{claim}{section}
\numberwithin{corollary}{section}
\numberwithin{definition}{section}
\numberwithin{observation}{section}
\numberwithin{proposition}{section}

\newcommand{\Prb}{\mathbb{P}}
\newcommand{\Prob}{\Prb}

\newcommand{\dist}{\mathrm{dist}}

\newcommand{\eps}{\varepsilon}

\newcommand{\Oh}{\mathcal{O}}

\newcommand{\Otilde}{\widetilde{\Oh}}
\newcommand{\Ot}{\Otilde}

\newcommand{\ff}{\mathcal{F}}
\newcommand{\nat}{\mathbb{N}}

\newcommand{\Ff}{\mathcal{F}}
\newcommand{\Dd}{\mathcal{D}}

\newcommand{\Ss}{\mathcal{S}}
\newcommand{\Cc}{\mathcal{C}}

\newcommand{\Gg}{\mathcal{G}}
\newcommand{\Pp}{\mathcal{P}}
\newcommand{\Kk}{\mathcal{K}}

\newcommand{\real}{\mathbb{R}_{+}}

\newcommand{\prob}[2]{\mathbb{P}_{#2}\left[ #1 \right]}
\newcommand{\Ex}[1]{\mathbb{E}\left[ #1 \right]}

\newcommand{\Vor}{\mathrm{Vor}}
\newcommand{\Diag}{\mathrm{Diag}}
\newcommand{\Rad}{\mathrm{Rad}}

\newcommand{\Allowed}{\mathrm{Allowed}}

\newcommand{\bal}{\mathrm{bal}}
\newcommand{\cmplx}{\mathrm{len}}

\newcommand{\enc}{\mathrm{Enc}}
\newcommand{\exc}{\mathrm{Exc}}

\renewcommand{\leq}{\leqslant}
\renewcommand{\geq}{\geqslant}
\renewcommand{\le}{\leqslant}
\renewcommand{\ge}{\geqslant}
\renewcommand{\setminus}{-}
\newcommand{\GoToAppendix}{$\bigstar$\xspace}

\title{A polynomial-time $\text{OPT}^\eps$-approximation algorithm for maximum independent set of connected subgraphs in a planar graph}
\date{}

\author{
    Jana Cslovjecsek\footnote{EPFL, Switzerland, \textsf{cslovj@gmail.com}.}
    \and
    Micha\l{} Pilipczuk\footnote{Institute of Informatics, University of
    Warsaw, Poland, \textsf{michal.pilipczuk@mimuw.edu.pl}. This work is a part of
    the project BOBR that has received funding from the European
    Research Council (ERC) under the European Union's Horizon 2020 research and
    innovation programme (grant agreement No 948057).}
    \and
    Karol W\k{e}grzycki\footnote{Saarland University and Max Planck Institute for Informatics,
        Saarbr\"ucken, Germany, \textsf{wegrzycki@cs.uni-saarland.de}. 
    This work is part of the project TIPEA that has
    received funding from the European Research Council (ERC) under the European Union's Horizon
    2020 research and innovation programme (grant agreement No 850979).}
}

\begin{document}

\maketitle

\thispagestyle{empty}
\begin{abstract}
    \input{chapters/abstract}
\end{abstract}

\begin{picture}(0,0)
\put(462,-250)
{\hbox{\includegraphics[width=40px]{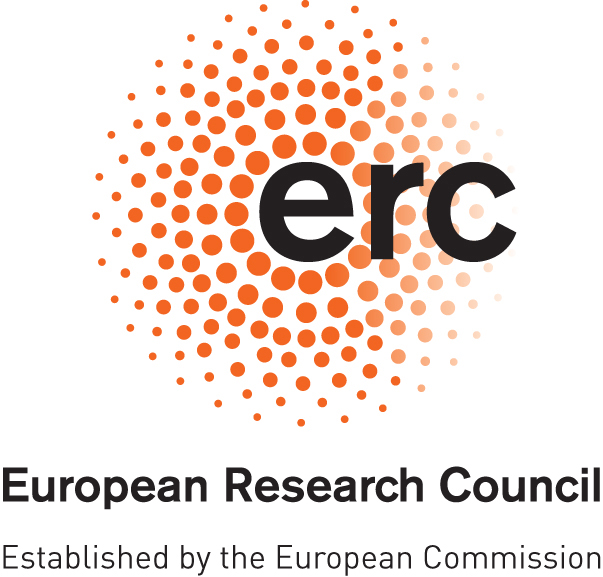}}}
\put(452,-310)
{\hbox{\includegraphics[width=60px]{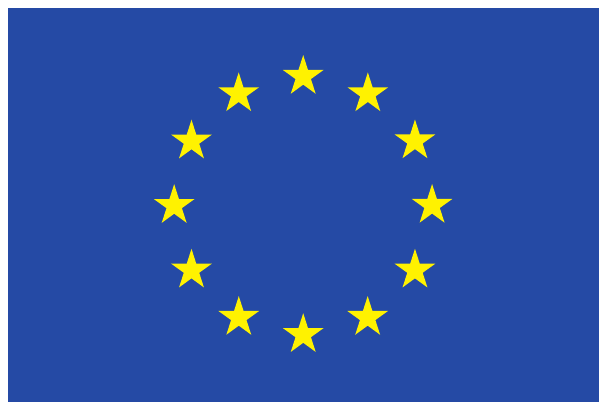}}}
\end{picture}

\clearpage
\setcounter{page}{1}

\input{chapters/introduction}

\input{chapters/basic-toolbox}

\input{chapters/sampling}
\input{chapters/algorithm}

\paragraph*{Acknowledgements.} The results presented in this paper were obtained during the trimester on Discrete Optimization at Hausdorff Research Institute for Mathematics (HIM) in Bonn, Germany. 
We are thankful for the possibility of working in the stimulating and creative research environment at HIM.

\bibliographystyle{abbrv}
\bibliography{bib}

\begin{appendix}
\input{chapters/appendix}

\input{chapters/appendix-sep}
\end{appendix}

\end{document}

%% file: chapters/abstract.tex
In the \textsc{Maximum Independent Set of Objects} problem, we are given an $n$-vertex planar graph
$G$ and a family $\mathcal{D}$ of $N$ \emph{objects}, where each object is a connected subgraph of
$G$. The task is to find a subfamily $\mathcal{F} \subseteq \mathcal{D}$ of
maximum cardinality that consists of pairwise disjoint objects. This problem is
$\mathsf{NP}$-hard and is equivalent to the problem of finding the maximum number of pairwise disjoint polygons in a given family of polygons in the plane.

As shown by Adamaszek et al. (J. ACM~'19), the problem admits a \emph{quasi-polynomial time
approximation scheme} (QPTAS): 
a $(1-\varepsilon)$-approximation algorithm whose running time is bounded by \mbox{$2^{\mathrm{poly}(\log(N),1/\epsilon)}
\cdot n^{\mathcal{O}(1)}$}. Nevertheless, to the best of our knowledge,
in the polynomial-time regime only the
trivial $\mathcal{O}(N)$-approximation is known for the problem in full generality. In the restricted setting where the objects are
pseudolines in the plane, Fox and Pach (SODA~'11) gave an
$N^{\varepsilon}$-approximation algorithm with running time $N^{2^{\tilde{\mathcal{O}}(1/\varepsilon)}}$, for any $\varepsilon>0$.

In this work, we present an $\text{OPT}^{\varepsilon}$-approximation algorithm for the problem that runs
in time $N^{\tilde{\mathcal{O}}(1/\varepsilon^2)} n^{\mathcal{O}(1)}$, for any $\varepsilon>0$, thus improving upon the result of Fox and Pach both in terms of generality and in terms of the running time. Our approach combines the methodology of \emph{Voronoi separators}, introduced by Marx and Pilipczuk (TALG~'22), with a new analysis of the approximation~factor.

%% file: chapters/introduction.tex
\section{Introduction}

\textsc{Maximum Independent Set} is a fundamental optimization problem defined as follows: Given a graph~$G$,  find a vertex subset $I \subseteq V(G)$ of maximum
cardinality such that no two vertices in $I$ are adjacent. The problem is $\mathsf{NP}$-complete and in full generality, very hard to approximate: as proved by H\r{a}stad~\cite{haastad1999clique}, there is no $|V(G)|^{1-\eps}$-approximation algorithm running in polynomial time for any $\eps > 0$, assuming $\mathsf{NP}\neq \mathsf{ZPP}$.

Therefore, special cases of the problem were investigated with the hope of finding more positive results. In this work we study the \textsc{Maximum Independent
Set of Objects} problem, which boils down to \textsc{Independent Set} in
intersection graphs of families of connected subgraphs of a planar graph. More
precisely, the input consists of a planar $n$-vertex graph $G$ and a set $\Dd$
of $N$ connected subgraphs of $G$, further called {\em{objects}}. The task is to find a maximum size subset $\Ff \subseteq \Dd$ such that objects in $\Ff$ are pairwise~disjoint. 

As observed in~\cite{esa18}, the \textsc{Maximum Independent
Set of Objects} problem is equivalent to the geometric problem of finding the
maximum number of pairwise disjoint polygons in a given family of polygons in the plane. The latter formulation is one of the
oldest problems in computational geometry, with many applications, for instance in
cartography~\cite{application-cartography-1,application-cartography-2,application-cartography-3},
interval scheduling~\cite{interval-scheduling}, chip
manufacturing~\cite{application-chips}, and even cellular
networks~\cite{application-cellular}.

The problem has also received significant attention from the theoretical community. For example, when objects are
\emph{fat} (e.g., disks or squares), polynomial-time approximation schemes
(PTASes) are known. The standard approach~\cite{ErlebachJS05} relies on the
quadtree combined with dynamic programming by Arora~\cite{Arora98}. Other
approaches include the usage of the planar separator theorem (see~\cite{Chan03}) or
local search strategies~\cite{chan2009approximation}; however, all of these
developments crucially rely on the geometric assumption of the fatness of the
objects. Another geometric setting that has been heavily studied is when all objects are axis-parallel rectangles. No PTAS is known for this case, while constant-factor approximation algorithms were proposed only very recently~\cite{galvez22,mitchell21}. 

Finally, Fox and Pach~\cite{FoxP11} considered the case when the considered
objects form a family of pseudolines. More generally, they assumed the objects
to be non-self-crossing curves in the plane such that any two of them intersect
in at most $k$ points, where $k$ is a constant. For this case, they gave an
$N^\eps$-approximation algorithm with a running time $N^{2^{\Ot_k(1/\eps)}}$;
the running time is independent of $N$, as their algorithm does not need to rely
on any geometric representation and can be given on input only the intersection
graph of the objects. The result of Fox and Pach improved the earlier work of
Agarwal and Mustafa~\cite{AgarwalM06}, who gave an $N^{1/2+o(1)}$-approximation
algorithm with a running time $\Oh(N^3)$ for the special case of segments in the plane.

Surprisingly, if one is willing to use time slightly higher than polynomial,
then much better approximation factors can be achieved. A recent line of work due to Adamaszek, Har-Peled, and Wiese~\cite{adamaszek19,adamaszek13,adamaszek-soda14} gave a \emph{quasipolynomial-time approximation scheme}
(QPTAS) for \textsc{Maximum Independent Set of Objects} in full generality with
a running time $2^{\mathrm{poly}(\log(N),\eps)} \cdot n^{\Oh(1)}$. See also the work of Pilipczuk et al.~\cite{esa18} for a somewhat streamlined presentation in the planar setting. Despite that, to the best of our knowledge, for the general \textsc{Maximum Independent Set of Objects} problem, the best approximation factor known to be achievable in polynomial time is the trivial $N$-approximation: pack one object.
%
%
%

\paragraph*{Our contribution.} In this work, we give a polynomial-time approximation algorithm for the \textsc{Maximum Independent Set of Objects} problem in full generality with a non-trivial approximation factor.

\begin{theorem}\label{thm:main}
    For every $\eps > 0$, \textsc{Maximum Independent Set of
    Objects} admits an $\text{OPT}^{\eps}$-approximation algorithm running in time
    $N^{\Ot(1/\eps^2)} \cdot n^{\Oh(1)}$, where $N$
    is the number of input objects, $n$ is the vertex count of the input graph, and $\text{OPT}$ is the cardinality of an optimum solution.
\end{theorem}

Observe that \cref{thm:main} improves the result of Fox and Pach~\cite{FoxP11} in terms of the running time ($N^{\Ot(1/\eps^2)} \cdot n^{\Oh(1)}$ instead of $N^{2^{\Ot_k(1/\eps)}}$), in terms of the approximation factor ($\text{OPT}^{\eps}$ instead of $N^{\eps}$), and in terms of generality (we do not assume the objects to be curves that pairwise intersect in at most $k$ points). The caveat is that our algorithm has to be provided a representation of the input consisting of a planar graph $G$ and a family of objects $\Dd$ in $G$, while the algorithm of Fox and Pach can rely only on the intersection graph of $\Dd$.

\paragraph*{Our techniques.}
The backbone of our proof of \cref{thm:main} is the separator approach used for
the design of QPTASes~\cite{adamaszek19,adamaszek13,adamaszek-soda14,esa18}. The
crux of this approach lies in the existence of a {\em{light separator}}, which
splits the solution in a balanced way and interacts with only a small fraction
of the solution. The following informal statement sketches the main properties of such a separator. For a formalization, see e.g.~\cite[Lemma~3]{esa18}.

\begin{lemma}[Light Separator Lemma]\label{lem:light-sep-lemma}
 Let $(G,\Dd)$ be an instance of  \textsc{Maximum Independent Set of Objects}, $\Ff\subseteq \Dd$ be any independent set of objects, and $\delta>0$ be an accuracy parameter. Then there exists a separator $S$ in $G$ of complexity $\mathrm{poly}(1/\delta)$ such that $S$ intersects at most $\delta|\Ff|$ objects in $\Ff$ and $S$ breaks $\Ff$ into subsets of size at most $\frac{2}{3}|\Ff|$.
\end{lemma}

In the above statement we did not specify the precise notion of a separator that
we use, the definition of its complexity, and what it means for $\Ff$ to be split
by a separator. Following Pilipczuk et al.~\cite{esa18}, in this work, we use the
notion of a {\em{radial separator}}, which is essentially a cycle separator in
the radial graph of the Voronoi diagram induced by the solution $\Ff$.
Importantly, there are only $|\Dd|^{\Oh(\ell)}$ candidates for a radial separator of complexity at most $\ell$, hence the separator $S$ whose existence is asserted by \cref{lem:light-sep-lemma} can be guessed among $|\Dd|^{\mathrm{poly}(1/\delta)}$ candidates.

The existence of balanced light separators suggests a simple recursive strategy: guess a balanced light separator $S$ of the optimum solution $\Ff$, sacrifice objects intersected by $S$, and recurse on subproblems based on $S$. This strategy is applied until the optimum solution size is $\mathrm{poly}(1/\eps)$, in which case an optimum solution can be found by brute-force. Clearly, the recursion depth is $\Oh(\log N)$ and it is not hard to see that every level of the recursion may incur a loss of a $\delta$-fraction of the optimum solution $\Ff$. Consequently, one needs to set $\delta=\frac{\eps}{\log N}$ to make sure that the total loss is bounded by $\eps|\Ff|$, and hence all the separators within the algorithm need to be guessed among quasi-polynomially many candidates. This is the fundamental reason why the algorithms of~\cite{adamaszek19,adamaszek13,adamaszek-soda14,esa18} are QPTASes, and not PTASes.

Our proof of \cref{thm:main} relies on the same concept of balanced light
separators provided by \cref{lem:light-sep-lemma}, but both the algorithm and the
analysis of the approximation factor are very different. In essence, we set
$\delta=\mathrm{poly}(1/\eps)$ instead of $\delta=\mathrm{poly}(1/\eps,\log N)$
and use balanced light separators for such $\delta$ to form a hierarchical
decomposition of the optimum solution $\Ff$ using so-called {\em{Swiss-cheese
separators}} of total complexity $\mathrm{poly}(1/\eps)$; each Swiss-cheese
separator consists of several radial separators. Importantly, at every step of
the decomposition, we use {\em{two}} radial separators to split in a balanced
way both the currently considered part of $\Ff$ and the radial separators ``stacked'' by the previous levels of the recursion. This way, the total complexity of each Swiss-cheese separator present in the decomposition is bounded by $\mathrm{poly}(1/\eps)$.
We note that this stacking problem did not really arise in the QPTASes of~\cite{adamaszek19,adamaszek13,adamaszek-soda14,esa18}, because there one can afford an additional $\log N$ factor in the total size of the separators stacked during the recursion, and the recursion tree has depth $\Oh(\log N)$ anyway.

Once the existence of a hierarchical decomposition as above is argued, we perform careful dynamic programming that guesses the decomposition ``on the fly'' and at the same time computes an approximate solution to the problem. The main novelty here lies in a new analysis of the approximation factor: we show that an $\mathrm{OPT}^\eps$-bound on the approximation factor can be proved by induction on the hierarchical decomposition, even without exploiting any bound on the depth of this decomposition.

\paragraph*{Organization.}
In \cref{sec:basic-toolbox} we recall the machinery of Voronoi diagrams and separators due to Marx and Pilipczuk~\cite{esa15}. \cref{sec:separator} introduces the notion of a Swiss-cheese separator and provides a balanced separator theorem for them. In \cref{sec:algorithm} we gather all the tools and prove \cref{thm:main}. Proofs of statements marked with \GoToAppendix are deferred
to~\cref{sec:appendix-a,sec:proof-separator}.

%% file: chapters/basic-toolbox.tex
\section{Basic toolbox}\label{sec:basic-toolbox}

In this section, we recall the basic toolbox introduced in~\cite{esa15,esa18}. Therefore, the text in this section is largely a paraphrase and an adaptation of the corresponding sections of~\cite{esa15,esa18}, and serves as an introduction to the methodology of Marx and Pilipczuk~\cite{esa15}. We remark that this methodology was also used by Cohen-Addad, Pilipczuk, and Pilipczuk~\cite{Cohen-AddadPP19} in the context of facility location problems in planar graphs.


Recall that in the {\sc{Maximum Independent Set of Objects}} problem the input
consists of a planar graph $G$ and a family of {\em{objects}} $\Dd$, where each
object is a connected subgraph of $G$. We fix such an instance for the remainder
of the paper. Further, we fix an embedding of $G$ in the sphere $\mathbb{S}^2$
and by triangulating every face, we assume that $G$ is a triangulation. For the
purpose of defining Voronoi diagrams on $G$, it will be convenient to impose a distance metric in $G$. For this, we assume that each edge $uv$ of $G$ is assigned a positive length, so that we may consider the shortest path metric in $G$, denoted $\dist(\cdot,\cdot)$. By perturbing the edge lengths slightly, we can make sure that all distances between pairs of vertices in $G$ are pairwise different; handling this perturbation increases the running times of our algorithms by only a polynomial factor. Note that thus, shortest paths are unique in $G$.

Let $\Ff\subseteq \Dd$ be an {\em{independent set}} in $\Dd$, that is, a
subfamily consisting of pairwise disjoint objects. We first define the
{\em{Voronoi partition}}: it is the partition of $V(G)$ into {\em{Voronoi
cells}} $\{\Vor_\Ff(p)\colon p\in \Ff\}$ such that $\Vor_\Ff(p)$ consists of
vertices $u$ of $G$ such that $\dist(u,p)<\dist(u,p')$ for all $p'\in \Ff$,
$p'\neq p$. Here, for a vertex $u$ and an object $p$ we denote
$\dist(u,p)\coloneqq \min_{v\in p} \dist(u,v)$. As the distances in $G$ are pairwise different, there are no ties.

Now comes the key definition from~\cite{esa15}, that of the {\em{Voronoi diagram}} of an independent set $\Ff\subseteq \Dd$. Such a diagram $\Diag_\Ff$ is constructed through a three-step process:
\begin{enumerate}[nosep,label=(\roman*)]
 \item For every $p\in \Ff$, fix an arbitrary spanning tree $T(p)$ of $p$. Then extend $T(p)$ to a spanning tree $\widehat{T}(p)$ of $G[\Vor_\Ff(p)]$ by adding, for every $u\in \Vor_\Ff(p)$, a shortest path from $u$ to $p$.
 \item Take the dual $G^\star$ of $G$ and remove from $G^\star$ all the edges dual to the edges of $\widehat{T}(p)$, for every $p\in \Ff$.
 \item From the obtained graph iteratively remove vertices of degree $1$ and
     bypass vertices of degree $2$, where {\em{bypassing}} a vertex of degree
     $2$ means replacing it together with the incident edges with a new edge connecting the neighbors. Once these operations are applied exhaustively, the multigraph obtained at the end is the Voronoi diagram $\Diag_\Ff$.
\end{enumerate}
As proved in~\cite[Lemmas 4.4, 4.5 and Section 4.4]{esa15}, the diagram $\Diag_\Ff$ is a $3$-regular connected multigraph with a sphere embedding naturally inherited from $G$ ($3$-regularity follows from the assumption that $G$ is triangulated). Moreover, $\Diag_\Ff$ has exactly $|\Ff|$ faces that bijectively correspond to the objects of $\Ff$, in the sense that each object $p\in \Ff$ is the unique object of $\Ff$ that is entirely contained in the face of $\Diag_\Ff$ corresponding to $p$. The vertices of $\Diag_\Ff$ are called {\em{branching points}}. Note that every branching point $f$ is actually a (triangular) face of the original graph $G$, and the duals of three edges of $f$ participate in the three different edges of $\Diag_\Ff$ that are incident with $f$.

\paragraph*{Singular faces.} Next, we recall the notion of {\em{singular faces}}, which characterize candidates for branching points in a diagram. The following is a paraphrase of the definition from~\cite{esa15}.

A face $f$ of $G$ is a {\em{singular face}} of {\em{type $1$}} for a triple of
disjoint objects $\Pp=\{p_1,p_2,p_3\}\subseteq \Dd$ if the vertices of $f$
belong to pairwise different cells of the Voronoi partition $\Vor_\Pp$. Next,
$f$ is a singular face of {\em{type $2$}} for such a triple $\Pp$ if two
vertices $u_1,u_2$ of $f$ belong to $\Vor_\Pp(p_1)$, the third vertex $u_3$ of
$f$ belongs to  $\Vor_\Pp(p_2)$, and the cycle formed by closing the $u_1$-$u_2$
path in $\widehat{T}(p_1)$ using the edge $u_1u_2$ of $f$ separates $p_2$ and
$p_3$ on $\mathbb{S}^2$. Finally, $f$ is a {\em{singular face}} of {\em{type
$3$}} for a quadruple of disjoint objects $\Pp=\{p_0,p_1,p_2,p_3\}\subseteq \Dd$ if
all vertices of $f$ belong to $\Vor_\Pp(p_0)$, but objects $p_1,p_2,p_3$ are
pairwise separated on the sphere by the union of the edges of $f$ and the minimal subtree of $\widehat{T}(p_0)$ that connects the vertices of $f$. See \cref{fig:fig8-in-esa15}.

\begin{figure}[ht!]
    \centering
    \begin{subfigure}[1]{0.25\textwidth}
    \def\svgwidth{\textwidth}
    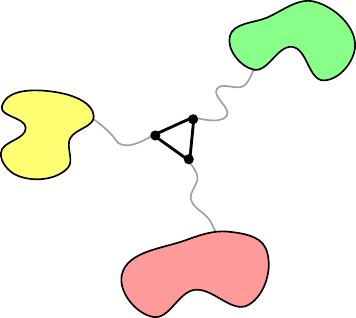
    \end{subfigure}
    \hspace{1cm}
    \begin{subfigure}[1]{0.25\textwidth}
    \def\svgwidth{\textwidth}
    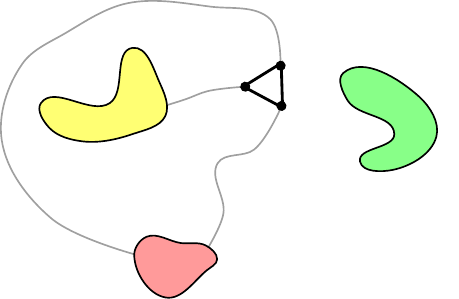
    \end{subfigure}
    \hspace{1cm}
    \begin{subfigure}[1]{0.25\textwidth}
    \def\svgwidth{\textwidth}
    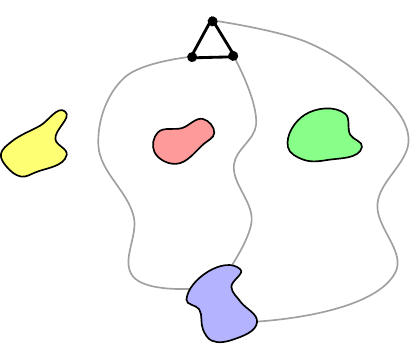
    \end{subfigure}
    
    \caption{Singular faces of types 1, 2 and 3. Relevant spokes depicted in gray. Figure reproduced from~\cite{esa15}, with the consent of the authors.}
    \label{fig:fig8-in-esa15}
\end{figure}

We will need the following two statements from~\cite{esa15}: there are few candidates for singular faces, and branching points of any diagram are always singular faces.

\begin{lemma}[Lemmas 4.8,4.9 and 4.10 in~\cite{esa15}]
    \label{lem:few-singular-faces}
    For each independent triple of objects $(p_1,p_2,p_3)$ there are at most two
    singular faces of type $1$ for $(p_1,p_2,p_3)$, and at most one singular
    face of type $2$ for $(p_1,p_2,p_3)$. For each independent quadruple of
    objects $(p_0,p_1,p_2,p_3)$, there is at most one singular face of type $3$
    for $(p_0,p_1,p_2,p_3)$.
\end{lemma}

\begin{lemma}[Lemma 4.12 in~\cite{esa15}]
    \label{lem:types-are-enough}
    Let $\Ff \subseteq \Dd$ be an independent set of objects, and let
    $\Diag_\Ff$ be the Voronoi diagram of $\Ff$. Then, every branching point of
    $\Diag_\ff$ is either a type-$1$ singular face for some triple of objects from $\Ff$, or a
    type-$2$ singular face for some triple of objects from $\Ff$, or a type-$3$
    singular face for some quadruple of objects from $\Ff$.
\end{lemma}

\paragraph*{Radial graph.}
Next, we present the {\em{radial graph}} $\Rad_\Ff$ of the diagram
$\Diag_\Ff$ (see~\cite{esa15,esa18}). This is the bipartite multigraph whose vertex set consists of the objects of $\Ff$ (which, recall, correspond to the faces of $\Diag_\Ff$) and the vertices of $\Diag_\Ff$ (that is, branching points). For every branching point $f$ and every vertex $u$ of~$f$, we create in $\Diag_\Ff$ one edge between $f$ and the unique object $p\in \Ff$ such that $u\in \Vor_\Ff(p)$. For convenience, this edge of $\Rad_\Ff$ will be {\em{labelled}} by the vertex $u$. Note that $\Rad_\Ff$ may contain two or more edges connecting the same object $p$ and a branching point $f$; this happens when $\Diag_\Ff$ is incident to one or more bridges. However, every branching point will have degree exactly $3$ in $\Rad_\Ff$.

Consider an edge $pf$ of $\Rad_\Ff$ with label $u\in f$. The {\em{spoke}} induced by $pf$ in the diagram $\Diag_\Ff$ is the shortest path $P$ from $u$ to the (nearest vertex of) object $p$. Note that $P$ is entirely contained in the tree $\widehat{T}(p)$, which means that it is also entirely contained in the face of $\Diag_\Ff$ corresponding to $p$. We can now use those spokes to define a natural embedding of $\Rad_\Ff$ in $\mathbb{S}^2$: take the subgraph of $G$ consisting of all branching points, all objects of $\Ff$, and all spokes of $\Diag_\Ff$, contract the branching points and the objects of $\Ff$ to single vertices, and contract the spokes of $\Diag_\Ff$ to single edges.

If $P$ is a spoke induced in $\Diag_\Ff$, corresponding to an edge $pf$ of $\Rad_\Ff$, and $q\in \Dd$, $q\neq p$, is another object, then we say that $q$ is {\em{in conflict}} with $P$ if there exists a vertex $v$ on $P$ such that $\dist(v,q)<\dist(v,p)$. Note that then we necessarily have $q\notin \Ff$, as $P$ is entirely contained in the cell $\Vor_\Ff(p)$. If $H$ is a subgraph of $\Rad_\Ff$, then we say that $q$ is \emph{in conflict} with $H$ if $q$ is in conflict with any spoke induced by an edge of~$H$.

In the sequel, we will often work with cycles in the radial graph $\Rad_\Ff$. Every such cycle $C$ consists of consecutive edges $p_1f_1,f_1p_2,p_2f_2,f_2p_3,\ldots,p_{\ell}f_\ell,f_\ell p_1$, where $p_1,\ldots,p_\ell$ are objects visited by $C$, $f_1,\ldots,f_\ell$ are branching points visited by $C$, and the length of $C$ is $2\ell$. For such a cycle $C$, we define a closed walk $\Gamma(C)$ in $G$ by concatenating in the natural order (indices behave cyclically):
\begin{itemize}[nosep]
 \item the spokes induced by the edges of $C$;
 \item for every object $p_i$, the unique path within $T(p_i)$ between the endpoints of the spokes induced by $p_if_i$ and $f_{i-1}p_i$ belonging to $p_i$; and
 \item for every branching point $f_i$, the edge of $f_i$ connecting the endpoints of the spokes induced by $p_if_i$ and $f_ip_{i+1}$ belonging to $f_i$.
\end{itemize}
It is easy to observe that $\Gamma(C)$ is indeed a closed walk in $G$. Moreover, the removal of $\Gamma(C)$ from $\mathbb{S}^2$ breaks $\mathbb{S}^2$ into two connected regions that naturally correspond to the two sides of the cycle $C$.

Finally, observe that for every edge $e=fg$ of $\Diag_\Ff$ there is a unique face of $\Rad_\Ff$ that contains $e$; this face is a $4$-cycle $pfqg$, where $p,q$
are objects of $\Ff$ corresponding to the faces of $\Diag_\Ff$ on both sides of~$e$. (Note that possibly $f=g$ or $p=q$.) The face under consideration will be called the {\em{diamond}} of $e$ and denoted $\diamondsuit_\Ff(e)$. Note that thus, the sphere $\mathbb{S}^2$ is divided into diamonds $\{\diamondsuit_\Ff(e)\colon e\in E(\Diag_\Ff)\}$. We will often implicitly identify the diamond $\diamondsuit_\Ff(e)$ with the connected region of $\Gamma(\diamondsuit_\Ff(e))$ that contains $e$.


\paragraph*{Radial separators.} We now introduce the main separator notion that
will be used in this work. Essentially, these separators correspond to the {\em{Voronoi separators}} of Marx and Pilipczuk~\cite{esa15}, but it will be convenient to set up the definitional layer a bit differently.

A \emph{radial separator} $S$ is a pair $(\Pp,C)$, where $\Pp\subseteq \Dd$ is
an independent set of objects and $C$ is a cycle in $\Rad_\Pp$. The \emph{length} of $S$ is the length of $C$, and we denote $\Gamma(S)\coloneqq \Gamma(C)$. 

We say that a radial separator $S = (\Pp,C)$ is \emph{compatible} with an
independent set $\Ff \subseteq \Dd$ if $\Pp \subseteq \Ff$ and each spoke of
$\Diag_\Pp$ that is incident with a branching point on $C$ is not in conflict
with any~$q\in\Ff \setminus \Pp$.  We will use the following claim about
persistence of branching points under addition of non-conflicting objects.

\begin{lemma}\label{lem:persistence}
 Let $\Pp\subseteq\Ff$ be independent sets of objects in $\Dd$. Further, let $f$
 be a branching point in $\Diag_\Pp$ such that none of the three spokes incident
 to $f$ in $\Diag_\Pp$ is in conflict with any of the objects of $\Ff\setminus
 \Pp$. Then $f$ is also a branching point in $\Diag_\Ff$ and all spokes of
 $\Diag_\Pp$ incident with $f$ persist as spokes of $\Diag_\Ff$ incident with
 $f$.
\end{lemma}
\begin{proof}
 Since the spokes of $\Diag_\Pp$ incident with $f$ are not in conflict with any
 object of $\Ff\setminus \Pp$, they persist as shortest paths within the cells
 of the same objects in $\Vor_\Ff$, as compared to $\Vor_\Pp$. In particular,
 every vertex of $f$ is in the cell of the same object in $\Vor_\Ff$ as in
 $\Vor_\Pp$. Now, consider any two distinct vertices $u,v$ of $f$. If $u,v$
 belong to different cells in $\Vor_\Pp$, equivalently in $\Vor_\Ff$, then the
 dual of the edge $uv$ is for sure not removed in the iterative removal of
 vertices of degree $1$ in the construction of $\Diag_\Ff$, because in this
 process we remove only duals of edges connecting vertices from the same cell.
 On the other hand, if $u,v$ belong to the same cell in $\Diag_\Pp$, say of
 object $p$, then the cycle $L$ obtained by concatenating the edge $uv$ of $f$
 with the $u$-$v$ path in $\widehat{T}(p)$ must separate some two objects
 of~$\Pp$, hence it separates some two objects of $\Ff$ as well due to
 $\Ff\supseteq \Pp$. In the construction of $\Vor_\Ff$, again $uv$ is the only
 edge on the cycle $L$ that does not belong to $\widehat{T}(p)$. As $L$
 separates two objects of $\Ff$, it is easy to see that in the iterative removal
 of vertices of degree $1$ in the construction of $\Diag_\Ff$, the dual of the
 edge $uv$ also could not be removed. All in all, none of the duals of the edges
 of the face $f$ was removed in the construction of $\Diag_\Ff$, which means
 that $f$ remains a vertex of degree $3$ in $\Diag_\Ff$, hence a branching
 point. Then the spokes incident to $f$ in $\Diag_\Pp$ also remain spokes
 incident to $f$ in $\Diag_\Ff$, because the vertices of $f$ are in the cells of
 the same objects in $\Vor_\Ff$ as in $\Vor_\Pp$.
\end{proof}

We first verify that extending the family of objects in a radial separator to a compatible superset still yields a radial separator. This is because, by definition, no spoke defining the separator can cease to be a spoke when adding the additional objects.

\begin{lemma}\label{lem:radial-sep-comps}
    Let $S = (\Pp,C)$ be a radial separator and suppose
    $S$ is compatible with $\Ff \subseteq \Dd$. Then 
    each branching point on $C$ is also a branching point on $\Diag_\Ff$.
    Moreover, $C$ is also a cycle in $\Rad_\Ff$, which means that $(\Ff,C)$ is
    also a radial separator.
\end{lemma}
\begin{proof}
    Let $K$ be the set of spokes induced in $\Diag_\Pp$ that are incident with the branching points visited by~$C$. Since $S$ is compatible with $\Ff$, it holds that $\Pp \subseteq
    \Ff$ and every spoke in $K$ is not in conflict with any object of $\Ff \setminus \Pp$. 
    By \cref{lem:persistence}, all branching points traversed by $C$ persist as branching points in $\Diag_\Ff$, and all spokes of $K$ persist as spokes connecting the same objects and branching points in $\Diag_\Ff$. The latter in particular applies to all spokes corresponding to the edges of $C$.
    We conclude that all objects, branching points, and edges between them that are traversed by $C$ in $\Rad_\Pp$ persist intact in $\Rad_\Ff$. So $C$ is also a cycle in $\Rad_\Ff$ and $(\Ff,C)$ is a radial separator.    
\end{proof}

Finally, we show that for the purpose of defining a radial separator $S$, one
can always restrict attention to a small set of objects relevant to the branching points visited by $S$.

\begin{lemma}\label{lem:radial-reduce-size}
    Let $(\Pp,C)$ be a radial separator of length $\ell$.  Then there exists $\Pp'\subseteq \Pp$ such that $|\Pp'|\leq \frac{5}{2} \ell$, $C$ is also a cycle in the radial graph $\Rad_{\Pp'}$, and $(\Pp',C)$ is a radial separator compatible with~$\Pp$.
\end{lemma}
\begin{proof}
    Let $\Pp'$ consist of:
    \begin{itemize}[nosep]
     \item the $\ell/2$ objects traversed by $C$; and
     \item for each of the $\ell/2$ branching points $f$ traversed by $C$, the set $\Pp_f\subseteq \Pp$ consisting of three or four objects of $\Pp$ that define $f$ in the sense of \cref{lem:types-are-enough}.
    \end{itemize}
    Thus, $|\Pp'|\leq \frac{5}{2}\ell$. 
    
    Consider any branching point $f$ of $\Diag_\Pp$ traversed by $C$.
    By the property provided by \cref{lem:types-are-enough}, $f$ is also a branching point in $\Diag_{\Pp_f}$, and the construction in the proof of \cref{lem:types-are-enough} in \cite{esa15} actually shows that the spokes of $\Diag_{\Pp_f}$ incident to $f$ are the same as the spokes of $\Diag_{\Pp}$ incident to $f$. Since $\Pp'\supseteq \Pp_f$ and none of the objects of $\Pp'\setminus \Pp_f$ is in conflict with any of those spokes (because they are spokes in $\Diag_{\Pp}$ and $\Pp'\subseteq \Pp$), by \cref{lem:persistence} we conclude that $f$ remains a branching point in $\Diag_{\Pp'}$, and all spokes of $\Diag_\Pp$ incident with $f$ are also spokes of $\Diag_{\Pp'}$ incident with $f$. This in particular applies to the two spokes corresponding to the edges of $C$ incident with $f$.

    Hence, all objects, branching points, and edges between them that are traversed by $C$ in $\Rad_\Pp$ persist intact in $\Rad_{\Pp'}$, and $C$ is also a cycle in $\Rad_{\Pp'}$. Finally, that all objects of $\Pp-\Pp'$ are not in conflict with the spokes of $\Diag_{\Pp'}$ incident with the branching points traversed by $C$ follows immediately from the construction, as these spokes are entirely contained in $\bigcup_{p\in \Pp'} \Vor_{\Pp}(p)$. Hence $(\Pp',C)$ is a radial separator compatible with $(\Pp,C)$.
\end{proof}

%% file: voronoi-figs/type1.pdf_tex
\begingroup%
  \makeatletter%
  \providecommand\color[2][]{%
    \errmessage{(Inkscape) Color is used for the text in Inkscape, but the package 'color.sty' is not loaded}%
    \renewcommand\color[2][]{}%
  }%
  \providecommand\transparent[1]{%
    \errmessage{(Inkscape) Transparency is used (non-zero) for the text in Inkscape, but the package 'transparent.sty' is not loaded}%
    \renewcommand\transparent[1]{}%
  }%
  \providecommand\rotatebox[2]{#2}%
  \ifx\svgwidth\undefined%
    \setlength{\unitlength}{170.85244253bp}%
    \ifx\svgscale\undefined%
      \relax%
    \else%
      \setlength{\unitlength}{\unitlength * \real{\svgscale}}%
    \fi%
  \else%
    \setlength{\unitlength}{\svgwidth}%
  \fi%
  \global\let\svgwidth\undefined%
  \global\let\svgscale\undefined%
  \makeatother%
  \begin{picture}(1,0.89333546)%
    \put(0,0){\includegraphics[width=\unitlength]{./voronoi-figs/type1.pdf}}%
    \put(0.10655534,0.53190193){\color[rgb]{0,0,0}\makebox(0,0)[lb]{\smash{{\tiny{$p_1$}}}}}%
    \put(0.51541559,0.13035243){\color[rgb]{0,0,0}\makebox(0,0)[lb]{\smash{{\tiny{$p_2$}}}}}%
    \put(0.818438,0.81573146){\color[rgb]{0,0,0}\makebox(0,0)[lb]{\smash{{\tiny{$p_3$}}}}}%
    \put(0.48685776,0.50176839){\color[rgb]{0,0,0}\makebox(0,0)[lb]{\smash{{\tiny{$f$}}}}}%
  \end{picture}%
\endgroup%

%% file: voronoi-figs/type2.pdf_tex
\begingroup%
  \makeatletter%
  \providecommand\color[2][]{%
    \errmessage{(Inkscape) Color is used for the text in Inkscape, but the package 'color.sty' is not loaded}%
    \renewcommand\color[2][]{}%
  }%
  \providecommand\transparent[1]{%
    \errmessage{(Inkscape) Transparency is used (non-zero) for the text in Inkscape, but the package 'transparent.sty' is not loaded}%
    \renewcommand\transparent[1]{}%
  }%
  \providecommand\rotatebox[2]{#2}%
  \ifx\svgwidth\undefined%
    \setlength{\unitlength}{220.23300391bp}%
    \ifx\svgscale\undefined%
      \relax%
    \else%
      \setlength{\unitlength}{\unitlength * \real{\svgscale}}%
    \fi%
  \else%
    \setlength{\unitlength}{\svgwidth}%
  \fi%
  \global\let\svgwidth\undefined%
  \global\let\svgscale\undefined%
  \makeatother%
  \begin{picture}(1,0.65079343)%
    \put(0,0){\includegraphics[width=\unitlength]{voronoi-figs/type2.pdf}}%
    \put(0.26264106,0.39242849){\color[rgb]{0,0,0}\makebox(0,0)[lb]{\smash{{\tiny{$p_1$}}}}}%
    \put(0.3373737,0.05266759){\color[rgb]{0,0,0}\makebox(0,0)[lb]{\smash{{\tiny{$p_2$}}}}}%
    \put(0.79437836,0.43595106){\color[rgb]{0,0,0}\makebox(0,0)[lb]{\smash{{\tiny{$p_3$}}}}}%
    \put(0.63039521,0.52273725){\color[rgb]{0,0,0}\makebox(0,0)[lb]{\smash{{\tiny{$u_3$}}}}}%
    \put(0.63458389,0.39356598){\color[rgb]{0,0,0}\makebox(0,0)[lb]{\smash{{\tiny{$u_2$}}}}}%
    \put(0.49477946,0.49471407){\color[rgb]{0,0,0}\makebox(0,0)[lb]{\smash{{\tiny{$u_1$}}}}}%
    \put(0.57414963,0.45614663){\color[rgb]{0,0,0}\makebox(0,0)[lb]{\smash{{\tiny{$f$}}}}}%
    \put(0.28314381,0.60174681){\color[rgb]{0,0,0}\makebox(0,0)[lb]{\smash{{\tiny{$$}}}}}%
    \put(0.39715724,0.46994017){\color[rgb]{0,0,0}\makebox(0,0)[lb]{\smash{{\tiny{$$}}}}}%
    \put(0.50197897,0.23790043){\color[rgb]{0,0,0}\makebox(0,0)[lb]{\smash{{\tiny{$$}}}}}%
  \end{picture}%
\endgroup%

%% file: voronoi-figs/type3.pdf_tex
\begingroup%
  \makeatletter%
  \providecommand\color[2][]{%
    \errmessage{(Inkscape) Color is used for the text in Inkscape, but the package 'color.sty' is not loaded}%
    \renewcommand\color[2][]{}%
  }%
  \providecommand\transparent[1]{%
    \errmessage{(Inkscape) Transparency is used (non-zero) for the text in Inkscape, but the package 'transparent.sty' is not loaded}%
    \renewcommand\transparent[1]{}%
  }%
  \providecommand\rotatebox[2]{#2}%
  \ifx\svgwidth\undefined%
    \setlength{\unitlength}{197.6913566bp}%
    \ifx\svgscale\undefined%
      \relax%
    \else%
      \setlength{\unitlength}{\unitlength * \real{\svgscale}}%
    \fi%
  \else%
    \setlength{\unitlength}{\svgwidth}%
  \fi%
  \global\let\svgwidth\undefined%
  \global\let\svgscale\undefined%
  \makeatother%
  \begin{picture}(1,0.83290974)%
    \put(0,0){\includegraphics[width=\unitlength]{voronoi-figs/type3.pdf}}%
    \put(0.05840118,0.45632354){\color[rgb]{0,0,0}\makebox(0,0)[lb]{\smash{{\tiny{$p_1$}}}}}%
    \put(0.40240919,0.47727461){\color[rgb]{0,0,0}\makebox(0,0)[lb]{\smash{{\tiny{$p_2$}}}}}%
    \put(0.50997988,0.08571743){\color[rgb]{0,0,0}\makebox(0,0)[lb]{\smash{{\tiny{$p_0$}}}}}%
    \put(0.75063279,0.50480875){\color[rgb]{0,0,0}\makebox(0,0)[lb]{\smash{{\tiny{$p_3$}}}}}%
    \put(0.39564952,0.71564729){\color[rgb]{0,0,0}\makebox(0,0)[lb]{\smash{{\tiny{$u_3$}}}}}%
    \put(0.49740504,0.81030945){\color[rgb]{0,0,0}\makebox(0,0)[lb]{\smash{{\tiny{$u_2$}}}}}%
    \put(0.59131746,0.67610676){\color[rgb]{0,0,0}\makebox(0,0)[lb]{\smash{{\tiny{$u_1$}}}}}%
    \put(0.4966553,0.71913506){\color[rgb]{0,0,0}\makebox(0,0)[lb]{\smash{{\tiny{$f$}}}}}%
    \put(0.32781053,0.32931546){\color[rgb]{0,0,0}\makebox(0,0)[lb]{\smash{{\tiny{$$}}}}}%
    \put(0.62126298,0.30176082){\color[rgb]{0,0,0}\makebox(0,0)[lb]{\smash{{\tiny{$$}}}}}%
    \put(0.84622204,0.29569495){\color[rgb]{0,0,0}\makebox(0,0)[lb]{\smash{{\tiny{$$}}}}}%
  \end{picture}%
\endgroup%

%% file: chapters/sampling.tex
\section{Swiss-cheese separators}\label{sec:separator}

The building blocks of our dynamic programming are captured through the
following~notion. Here, two cycles in $\mathbb{S}^2$ are {\em{non-crossing}} if
there exists a homotopy that deforms one cycle into the other without any
intersections.

\begin{figure}[ht!]
    \centering
    \includegraphics[width=0.6\textwidth]{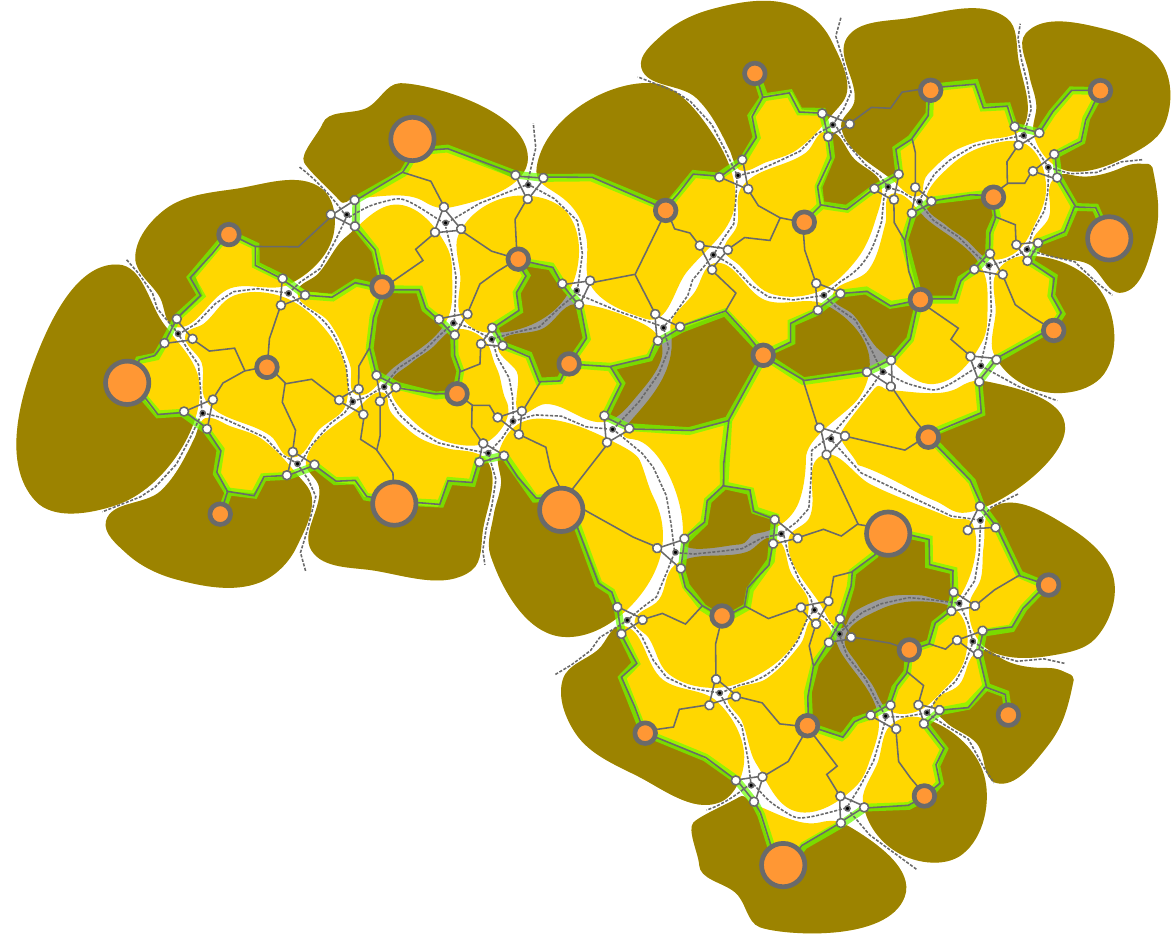}
    \caption{A schematic overview of a Swiss-cheese separator. The orange
    circles represent objects of $\Pp$. The triangular faces are the branching points of $\Diag_\Ff$. The yellow area (cheese mass) is the face $f$, while green parts (holes) are the regions excluded by the cycles of $\Kk$. The figure is built upon~\cite[Figure~1]{Cohen-AddadPP19}.}
    \label{fig:cheese}
\end{figure}

\begin{definition}[Swiss-cheese Separator] \label{def:our-separator}
    A \emph{Swiss-cheese separator} is a triple $(\Pp,\Kk,f)$ where $\Pp
    \subseteq \Dd$ is a set of pairwise disjoint objects, $\Kk$ is a set of
    edge-disjoint non-crossing cycles in $\Diag_\Pp$,
    and $f$ is a face of $\bigcup_{C\in \Kk} C$ (i.e., the subgraph of $\Diag_\Pp$ consisting of the union of all the vertices and edges of all the cycles in $\Kk$) whose boundary contains every cycle of $\Kk$.
    The {\em{complexity}} of the Swiss-cheese separator $(\Pp,\Kk,f)$ is
    the total number of edges of cycles in $\Kk$. A Swiss-cheese separator is {\em{ripe}} if $|\Pp|\leq \frac{5}{2}\ell$, where $\ell$ is its complexity.
\end{definition}

In the context of a Swiss-cheese separator $(\Pp,\Kk,C)$, we will often identify
the face $f$ with the corresponding region of $\mathbb{S}^2\setminus
\bigcup_{C\in \Kk}\Gamma(C)$, with the natural meaning (see~\cref{fig:cheese}).

Note that if $(\Pp,\Kk,C)$ is a Swiss-cheese separator, then 
for every $C \in \Kk$, the pair $(\Pp,C)$ is a radial separator.
The next definition is adapted from~\cite{esa18}.  For a  $C \in \Kk$, we say that an object $p \in \Dd$ is
{\em{banned}} by $C$ if $p$ is in conflict with any spoke of $\Diag_\Pp$
incident to any branching point traversed by $C$. We define 
\begin{displaymath}
    \Allowed(\Pp,\Kk,f)\coloneqq \left\{\,p\in \Dd\quad\Big|\quad p\subseteq
    f,\ p \textrm{ is not banned by any }C\in \Kk \textrm{ and }\Pp\cup \{p\}\textrm{ is independent} \right\}.
\end{displaymath}
In other words, $\Allowed(\Pp,\Kk,f)$ comprises all objects
entirely contained in the region $f$ that are not banned by the radial
separators $(\Pp,C)$, for any $C \in \Kk$, and are disjoint with the objects of $\Pp$. The point of this definition is that for any $C\in \Kk$ and $p\in \Allowed(\Pp,\Kk,f)$, $p$ is not in conflict with $C$.

Now, let us introduce the following \emph{size} parameter:
$$s \coloneqq s(\eps) = \frac{10^{10} \ln^2(1/\eps)}{\eps^2}.$$ 
The main goal of this section is to prove the following separator theorem.

\begin{theorem}[Swiss-cheese Separator Theorem]\label{lem:nice-separator}
    Let $\eps>0$ be a small enough parameter and let 
    $(\Pp, \Kk,f)$ be a ripe Swiss-cheese separator of complexity at most
    $s(\eps)$. Let $\Ff \subseteq \Allowed(\Pp,\Kk,f)$ be a set of pairwise disjoint
    objects such that $|\Ff| \ge s(\eps)$. Then there exists a set of disjoint objects $\Ss$ with $\Pp\subseteq \Ss\subseteq \Pp\cup \Ff$ and $|\Ss\setminus\Pp|\leq 2s(\eps)$ and two cycles $C_1,C_2$ in $\Rad_\Ss$ such that  
    \begin{enumerate}[label=\textbf{(\alph*)},nosep]
    \item at most $\eps |\Ff|$ objects of $\Ff$ are in conflict with $C_1$ or $C_2$.\label{enum:conflict}
    \end{enumerate}
    Further, let $R$ be a subgraph of $\Rad_{\Ss}$ consisting of all the edges and vertices of cycles corresponding to the cycles of $\Kk \cup \{C_1,C_2\}$ (where cycles of $\Kk$ are also cycles in $\Rad_\Ss$ by  \cref{lem:radial-sep-comps}). Then
    \begin{enumerate}[label=\textbf{(\alph*)},nosep,resume]
        \item for every face $g$ of $R$ satisfying $g\subseteq f$, the region of $\mathbb{S}^2-\bigcup_{C\in \Kk\cup \{C_1,C_2\}} \Gamma(C)$ corresponding to $g$ contains at most $\frac{3}{4} |\Ff|$ objects of $\Ff$, and \label{enum:constant-frac}
        \item the boundary of every face $g$ of $R$ satisfying $g\subseteq f$
        can be partitioned into edge-disjoint non-crossing cycles in $\Rad_{\Ss}$ of
        total length at most $s(\eps)$, and therefore $g$ can be represented as some ripe Swiss-cheese
        separator $(\Pp_g,\Kk_g,g)$ of complexity at most $s(\eps)$.\label{enum:partition}
    \end{enumerate}
\end{theorem}

The rest of this section is dedicated to the proof of~\cref{lem:nice-separator}. 
Therefore, from now on we fix the ripe Swiss-cheese separator $(\Pp,\Kk,f)$ and the family $\Ff\subseteq \Allowed(\Pp,\Kk,f)$.
Let $\ell$ be the complexity of $(\Pp,\Kk,f)$; we have $\ell\leq s\leq |\Ff|$. Also, recall that
$\Pp \subseteq \Dd$ and, as $(\Pp,\Kk,f)$ is ripe, also $|\Pp|\le \frac{5}{2}\ell\leq \frac{5}{2}s$. We will assume that $\eps<10^{-100}$, so that we also have $s>10^{10}$.

\subsection{Sampling}\label{sec:sampling}

The first step is to construct $\Ss$.
Similarly to~\cite{adamaszek19,Har-Peled14,esa18}, we do this by sampling objects from $\Ff$. More precisely, first we include in $\Ss$ every member of $\Pp$.
Then, for every object $p \in \Ff$, we include $p$ in $\Ss$ independently with
probability $\lambda \coloneqq \frac{s}{|\Ff|}$. Note that $\lambda \leq 1$, as
we assumed that $|\Ff| \ge s$.  This concludes the construction of $\Ss$. 
Our first observation is that with high probability, the size of $\Ss\setminus
\Pp$ is $\Theta(s)$. The proof is a standard application of Chernoff's inequality.

\begin{claim}\label{lem:sample-size}
    With probability at least $\frac{3}{4}$ it holds that $s/2 \le |\Ss
    \setminus \Pp| \le 2s$.
\end{claim}
\begin{proof}
    Let $X$ be the random variable equal to the cardinality of $\Ss \setminus
    \Pp$. Hence $X = \sum_{p \in \Ff} X_p$, where
    $X_p$ are random variables that take value $1$ if $p \in \Ff$ is included in $\Ss$ and $0$ otherwise. Note that $X_p$ are independent identically distributed Bernoulli variables with success probability $\lambda$, and $$\Ex{X} = \sum_{p \in \Ff} \Ex{X_p} =
    |\Ff| \cdot \lambda  = s.$$
    Hence, by Chernoff inequality we conclude that
    $$\prob{X<s/2}{}\leq e^{-s/8}\qquad\textrm{and}\qquad \prob{X>2s}{}\leq e^{-s/3}.$$
    As we assumed that $s>100$, the probability that any of the above events happens is bounded by $\frac{1}{4}$.
\end{proof}


Next, we analyze the properties of the set $\Ss$. This will generally follow the
framework of~\cite[Section~2.2]{esa18} with the caveat that the objects of $\Pp$
are always included in $\Ss$. We start by arguing that with high probability,
every spoke induced in $\Diag_\Ss$ is in conflict with only a few objects in $\Ff$. To measure this, we need the following definition: we say that a spoke of $\Diag_\Ss$ is \emph{heavy} if it is in conflict with at least
    $$\eta \coloneqq \frac{100 \ln(1/\eps)}{s} \cdot |\Ff|\qquad \textrm{objects in }\Ff.$$
The following statement guarantees that with high probability no spoke is heavy.

\begin{lemma}[\GoToAppendix]\label{lem:no-heavy-spoke}
    With probability at least $\frac{3}{4}$, all spokes induced in the diagram $\Diag_\Ss$ are not heavy.
\end{lemma}

The proof of~\cref{lem:no-heavy-spoke} is a direct adaptation of the proof
of~\cite[Claim 2]{esa18} to our setting. The main difference is that here we need to take
objects of $\Pp$ into account. 

Next, we focus our attention on diamonds, which, recall, are faces of the radial graph
$\Rad_\Ss$.
The \emph{weight} of a diamond $\diamondsuit_\Ss(e)$ with respect to $\Ff$ is the
number of objects of $\Ff$ that intersect
$\diamondsuit_\Ss(e)$, excluding the (at most) two objects of $\Ss$ that define
$\diamondsuit_\Ss(e)$. Note that all the objects contributing to the weight of
$\diamondsuit_\Ss(e)$ do not belong to $\Ss$. This leads to the following definition: call a diamond $\diamondsuit_\Ss(e)$ in $\Diag_\Ss$ {\em{heavy}} if its weight with respect to $\Ff$ is
larger than $\eta$. 
The next statement ensures that with high probability, in $\Diag_\Ss$ one does not observe any heavy
diamonds. The estimation of this probability basically follows the same
framework as the proof of~\cref{lem:no-heavy-spoke}, and is a direct adaptation of the proof
of~\cite[Claim 3]{esa18}.

\begin{lemma}[\GoToAppendix]\label{lem:no-heavy-diamond}
    With probability at least $\frac{3}{4}$, there is no heavy diamond in $\Diag_\Ss$.
\end{lemma}

\subsection{Finding cycles}

We proceed to finding suitable cycles $C_1,C_2$ in $\Rad_\Ss$. These cycles should be balanced with respect to two
different weight functions: one reflecting the density of $\Ff$ in different
parts of the diagram $\Diag_\Ss$, and second corresponding to the placement of the initial Swiss-cheese
separator $(\Pp,\Kk,f)$, as we wish to split this separator in a balanced way. To find the cycles, we will use the following statement. 

\begin{restatable}{theorem}{separator}{\normalfont{(\GoToAppendix)}}\label{thm:separator}
    Let $H$ be the Voronoi diagram of some set $\Ss$ of disjoint disjoint objects in a graph $G$ embedded in the sphere $\mathbb{S}^2$. Let
    $\omega$ be an assignment of nonnegative weights to vertices of $H$ such
    that the total sum of the weights is equal to $1$. Then there exists a cycle $C$ in $\Rad_\Ss$ of length at most $\sqrt{18|V(H)|}$ such that the total weight of the vertices of $H$ lying strictly
    inside any of the two regions of $\mathbb{S}^2-C$ is at most~$2/3$.
\end{restatable}

The proof of~\cref{thm:separator} follows from a rather straightforward adjustment of the arguments contained in~\cite[Sections~4.7 and~4.8]{esa15}, hence we defer it to~\cref{sec:proof-separator}.
With~\cref{thm:separator} in hand, we can now construct the two cycles
and prove~\cref{lem:nice-separator}. 

Recall that if we perform sampling as explained in
\cref{sec:sampling}, then by
\cref{lem:sample-size,lem:no-heavy-spoke,lem:no-heavy-diamond} with probability
at least $1-(\frac{1}{4}+\frac{1}{4}+\frac{1}{4})=\frac{1}{4}$ we get $\Ss$ with $\Pp\subseteq \Ss\subseteq \Pp\cup \Ff$ satisfying the~following:
\begin{itemize}[nosep]
 \item $s/2\leq |\Ss\setminus \Pp|\leq 2s$;
 \item no spoke in $\Diag_\Ss$ is heavy; and
 \item no diamond in $\Diag_\Ss$ is heavy.
\end{itemize}
Hence, there exists a sample $\Ss$ satisfying the above properties, which we fix from now on. Let $H\coloneqq \Diag_\Ss$.
The two cycles $C_1,C_2$ will be obtained by applying \cref{thm:separator} to two different weight functions on~$H$. The first one will guarantee
balancedness with respect to the number of objects of $\Ff$. More precisely, let $\mu_\bal
\colon V(H) \to [0,1]$ be  defined as follows.  For every
object $p\in \Ff$, assign $p$ to an arbitrary diamond $\diamondsuit_\Ss(e)$,
where $e\in V(H)$, such that $p$ intersects $\diamondsuit_\Ss(e)$ (including its
boundary). Next, assign every diamond $\diamondsuit_{\Ss}(e)$ to an arbitrarily
chosen endpoint of $e$. Finally, for $f\in V(H)$ define $\mu_\bal(f)$ to be the
total number of objects assigned to $f$ through the composition of the two
assignments described above, divided by $|\Ff|$.
Clearly, $\sum_{f\in V(H)} \mu_\bal(f)=1$. Also, the absence of heavy diamonds implies the following.

\begin{claim}\label{claim:mu-bal}
     For every $f\in V(H)$, $\mu_\bal(f)\leq 10^{-5}\cdot \frac{\eps^2}{\ln (1/\eps)}$.
\end{claim}
\begin{proof}
 Every object of $\Ff$ assigned to $f$ must either be an object of $\Ss$
 that corresponds to one of the at most three faces of $\Diag_\Ss$ incident to $f$,
 or must intersect one of the three diamonds corresponding to the edges of
 $\Diag_\Ss$ incident with $f$. Since every diamond in $\Diag_\Ss$ is not heavy,
 the claim follows due to $\mu_\bal(f)\leq \frac{3\eta+3}{|\Ff|}\leq10^{-5}\cdot
 \frac{\eps^2}{\ln (1/\eps)}$.
\end{proof}
Similarly, we define the weight function $\mu_\cmplx \colon V(H) \to [0,1]$ by setting $\mu_\cmplx(f)=\frac{2}{\ell}$ if $f$ is a branching point traversed by any of the cycles from $\Kk$, and $\mu_\cmplx(f)=0$ otherwise. Observe that since branching points have degree $3$ in $\Rad_\Ss$ and cycles of $\Kk$ are edge-disjoint, every branching point $f$ participates in at most one cycle from $\Kk$. As the total length of cycles in $\Kk$ is $\ell$, we conclude that $\sum_{f\in V(H)} \mu_\cmplx(f)=1$. 

We now apply \cref{thm:separator} to the weight functions $\mu_\bal$ and
$\mu_\cmplx$, yielding cycles $C_1$ and $C_2$ in $\Diag_\Ss$.  First, we
estimate the size of $H=\Diag_\Ss$ and the lengths of the cycles $C_1$ and
$C_2$. First, we have $|\Ss|=|\Pp|+|\Ss-\Pp|\leq \frac{9}{2}s$. As $\Diag_\Ss$
is a connected $3$-regular sphere-embedded multigraph with $|\Ss|$ faces, from
Euler's formula it follows that $H$ has less than $2|\Ss|\leq 9s$ branching
points (see~\cite[Lemma~4.5]{esa15}). So $|V(H)|\leq 9s$. Now, by
\cref{thm:separator} both $C_1$ and $C_2$ have lengths bounded by $\sqrt{18\cdot
9s}=\sqrt{162s}$.


To conclude the proof of~\cref{lem:nice-separator}, we are left with verifying
properties \ref{enum:conflict}, \ref{enum:constant-frac}, and
\ref{enum:partition}.  

\begin{proof}[Proof of property~\ref{enum:conflict} in~\cref{lem:nice-separator}]
For property~\ref{enum:conflict}, we need to show that only $\eps |\Ff|$
of objects of $\Ff$ are in conflict with any of the cycles $C_1$ or $C_2$. Recall that
each spoke of $H$ is in conflict with at most $\eta$ objects of $\Ff$. Therefore, the total number of objects in conflict with $C_1$ or $C_2$ is bounded by
\begin{displaymath}
    2 \cdot \sqrt{162s}\cdot \eta \le \frac{ \sqrt{648 s} \cdot 100 \ln(1/\eps)}{s}\cdot |\Ff| \le 10^5\cdot \frac{\ln(1/\eps)}{\sqrt{s}}\cdot |\Ff|= \eps |\Ff|,
\end{displaymath}
as required.
\end{proof}

\begin{proof}[Proof of property~\ref{enum:constant-frac} in~\cref{lem:nice-separator}]
For property~\ref{enum:constant-frac}, let $A$ be any of the two regions into which $\Gamma(C_1)$ divides the sphere $\mathbb{S}^2$. Observe that every object $p$ of $\Ff$ that is entirely contained in $A$ contributes, in the definition of the weight function $\mu_{\bal}$, to the weight of a branching point that lies either inside $A$ or on the cycles $C_1$; this is because $p$ can intersect only diamonds contained in $A$. By~\cref{thm:separator}, the total $\mu_\bal$-weight of branching points lying on either side of the cycle $C_1$ (excluding branching points on $C_1$ itself) is at most $\frac{2}{3}$, while by \cref{claim:mu-bal}, the total weight of branching points lying on $C_1$ is bounded by
$$10^{-5}\cdot
 \frac{\eps^2}{\ln (1/\eps)}\cdot \sqrt{162s}=10^{-5}\cdot\sqrt{162}\cdot  10^5\cdot \eps \leq \frac{1}{12}.$$
Hence, the total number of objects $p$ entirely contained in $A$ is bounded by $(\frac{2}{3}+\frac{1}{12})|\Ff|=\frac{3}{4}|\Ff|$. Since for every face $g$ of $R$, the corresponding region of $\mathbb{S}^2\setminus \bigcup_{C\in \Kk\cup \{C_1,C_2\}} \Gamma(C)$ is entirely contained in one of the regions of $\mathbb{S}^2\setminus \Gamma(C_1)$, property~\ref{enum:constant-frac} follows.
\end{proof}

\begin{proof}[Proof of property~\ref{enum:partition} in~\cref{lem:nice-separator}]
    For property~\ref{enum:partition}, let $g \subseteq f$ be
    a face of $R$. As $g$ is a face of a subgraph of $H$, the boundary of $g$ can be partitioned into a set $\Kk_g$ consisting of non-crossing edge-disjoint cycles in $H$. Let us now estimate the total length of the cycles in $\Kk_g$.

    Recall that the cycle $C_2$ is balanced with respect to the
    weight function $\mu_\cmplx$ in the sense of \cref{thm:separator}. Since $g$ is
    entirely contained in one of the sides of $C_2$, say $A$, it follows that every
    branching point visited by any cycle of $\Kk_g$ is either a branching point
    visited by $C_2$ or it lies strictly inside $A$. Since the total length of
    cycles in $\Kk$ is $\ell$ and every branching point traversed by a cycle of
    $\Kk$ is incident to exactly two edges of cycles in $\Kk$, the total number of branching points traversed by the cycles of $\Kk$ is $\frac{\ell}{2}$. By \cref{thm:separator}, at most $\frac{2}{3}\cdot \frac{\ell}{2}=\frac{\ell}{3}$ of these branching points may lie strictly inside $A$. On the other hand, cycle $C_2$ has length at most 
    $\sqrt{162s}\leq \frac{s}{3}$, which means that there are at most $\frac{s}{6}$ branching points traversed by $C_2$. All in all, the total number of branching points traversed by any cycle of $\Kk_g$ is bounded by
    $$\frac{\ell}{2}+\frac{s}{6}\leq \frac{s}{3}+\frac{s}{6}=\frac{s}{2}.$$
    Now every branching point traversed by any cycle of $\Kk_g$ is incident to
    exactly two edges of cycles in $\Kk_g$. So the total number of edges on the cycles of $\Kk_g$ is bounded by $s$, as required.

    It remains to define $\Pp_g$ so that $(\Pp_g,\Kk_g,g)$ is a ripe Swiss-cheese
    separator. For this, it suffices to apply \cref{lem:radial-reduce-size} to every
    cycle of $\Kk_g$ and take the union of the obtained subsets of $\Ss$ (note
    that the compatibility of cycles in $\Kk_g$ holds). The fact
    that $(\Pp_g,\Kk_g,g)$ obtained in this manner is a Swiss-cheese separator follows from \cref{lem:radial-sep-comps}.
\end{proof}

%% file: chapters/algorithm.tex
\section{Algorithm}\label{sec:algorithm}
\newcommand{\LHS}{\mathrm{LHS}}
\newcommand{\opt}{\mathrm{opt}}
\newcommand{\Sep}{\mathsf{Sep}}
\newcommand{\Cand}{\mathsf{Cand}}
\newcommand{\VorCand}{\Cc}
\newcommand{\SmallSol}{\mathsf{SmallSol}}
\newcommand{\ApproxIS}{\mathsf{ApproxIS}}

In this section, we present the proof of~\cref{thm:main}. We write $s\coloneqq s(\eps)$ for brevity and denote $n\coloneqq |V(G)|$. Our algorithm
enumerates the set of all ripe Swiss-cheese separators of complexity $s$:
\begin{align*}
    \Sep & \coloneqq \{ (\Pp,\Kk,f)\ \colon\ (\Pp,\Kk,f) \text{ is a ripe Swiss-cheese separator of complexity at most } s\}.
\end{align*}
Throughout this section we will think of $f$ in the above as of a region of $\mathbb{S}^2\setminus \bigcup_{C\in \Kk} \Gamma(C)$, as explained below \cref{def:our-separator}.

\begin{lemma}\label{obs:enum}
    $\Sep$ has size at most $|\Dd|^{\Ot(1/\eps^2)}$ and can be enumerated in time $|\Dd|^{\Ot(1/\eps^2)} \cdot n^{\Oh(1)}$.
\end{lemma}
\begin{proof}
    Recall that if $(\Pp,\Kk,f)$ is a ripe Swiss-chees separator of complexity at most $s$, then $|\Pp|\leq \frac{5}{2}s$. Therefore, to describe $(\Pp,\Kk,f)$ it suffices to provide the following components:
    \begin{itemize}[nosep]
     \item An independent set $\Pp\subseteq \Dd$ consisting of $\frac{5}{2}s$ objects, for which there are $|\Dd|^{\Oh(s)}$ options.
     \item A set $\Kk$ consisting of edge-disjoint non-crossing cycles in $\Diag_\Pp$. Note that $\Kk$ is uniquely defined by the set of edges participating in the cycles of $\Kk$, and $\Diag_\Ss$ has $\Oh(s)$ edges, hence there are $2^{\Oh(s)}$ options for $\Kk$.
     \item A face $f$ of the subgraph of $\Diag_\Pp$ consisting of edges and vertices of the cycles $\Kk$, for which there are at most $\Oh(s)$ options.
    \end{itemize}
    For each candidate as described above, we can verify in polynomial time whether it indeed defines a ripe Swiss-cheese separator. Since $s=\Ot(1/\eps^2)$, the claim follows.
\end{proof}

Next, we define the order in which separators in $\Sep$ will be considered. For two separators $(\Pp_1,\Kk_1,f_1)$ and $(\Pp_2,\Kk_2,f_2)$ in $\Sep$, we
say that $(\Pp_1,\Kk_1,f_1) \preceq (\Pp_2,\Kk_2,f_2)$ if the region $f_1$ is contained in
the region $f_2$. Observe that $\preceq$ is a partial order on $\Sep$. We also
define the optima that will be sought by our dynamic programming procedure.

\begin{definition}
    We say that an independent set $\Ff \subseteq \Dd$ \emph{is inside}
    $(\Pp,\Kk,f)$ if $\Ff\subseteq \Allowed(\Pp,\Kk,f)$. The maximum size of such an
    independent set will be denoted by $\opt(\Pp,\Kk,f)$.
\end{definition}

Next, we prove that given a separator $(\Pp,\Kk,f)\in \Sep$, one can enumerate a
family of candidate splits of $(\Pp,\Kk,f)$ into several separators from $\Sep$
with strictly smaller regions, with the guarantee that among the splits there is
one where only small fraction of objects of $\Ff$ is lost. The proof is just an application of~\cref{lem:nice-separator}.

\begin{lemma}\label{lem:enumerate-splits}
 Given $(\Pp,\Kk,f)\in \Sep$, one can in time $|\Dd|^{\Ot(1/\eps^2)}\cdot n^{\Oh(1)}$ enumerate a family $\Cand(\Pp,\Kk,f)$ of size $|\Dd|^{\Ot(1/\eps^2)}$, such that:
 \begin{itemize}[nosep]
  \item each member of $\Cand(\Pp,\Kk,f)$ is a family $\VorCand\subseteq \Sep$
      such that for each $(\Pp',\Kk',f')\in \VorCand$, the region $f'$ is a
      strict subset of $f$, and the interiors of $f'$ for $(\Pp',\Kk',f')\in
      \VorCand$ are pairwise disjoint; and
  \item provided $\opt(\Pp,\Kk,f)\geq s$, there exists $\VorCand\in \Cand(\Pp,\Kk,f)$ such that $$\sum_{(\Pp',\Kk',f')\in \VorCand} \min\left(\frac{3}{4}\cdot \opt(\Pp,\Kk,f),\opt(\Pp',\Kk',f')\right)\geq (1-\eps)\cdot \opt(\Pp,\Kk,f).$$
 \end{itemize}
\end{lemma}
\begin{proof}
 We enumerate the family $\Cand(\Pp,\Kk,f)$ as follows. Iterate through all
 independent sets $\Ss^\circ\subseteq \Allowed(\Pp,\Kk,f)$ of size at most $2s$
 and all pairs of cycles $C_1,C_2$ in $\Rad_\Ss$, where $\Ss\coloneqq \Pp\cup
 \Ss^\circ$. For each such choice, let $R$ be the subgraph of $\Rad_\Ss$
 consisting of all vertices and edges traversed by the cycles of $\Kk\cup
 \{C_1,C_2\}$ (where the cycles of $\Kk$ persist as cycles in $\Rad_{\Ss}$ due
 to \cref{lem:radial-sep-comps}). For every face $g$ of $R$ satisfying
 $g\subseteq f$, construct a ripe Swiss-cheese separator $(\Pp_g,\Kk_g,g)$ by
 applying \cref{lem:radial-reduce-size} to $(\Ss,\Kk_g,g)$, where $\Kk_g$ is the
 decomposition of the boundary of $g$ into edge-disjoint non-crossing cycles
 in~$\Diag_\Ss$. Note that the compatibility is preserved. Next, we let
 $\VorCand$ be the set of separators $(\Pp_g,\Kk_g,g)$ as described above. We include $\VorCand$ in $\Cand(\Pp,\Kk,f)$ if it satisfies the first condition from the lemma statement: $\VorCand\subseteq \Sep$ (equivalently, every separator in $\VorCand$ has complexity at most~$s$) and $g\subsetneq f$ for all $(\Pp_g,\Kk_g,g)\in \VorCand$.
 
 The family $\Cand(\Pp,\Kk,f)$ constructed in this way will have size
 $|\Dd|^{\Oh(s)}$, which is $|\Dd|^{\Ot(1/\eps^2)}$ due to $s=\Ot(1/\eps^2)$. It
 can also be enumerated in time $|\Dd|^{\Ot(1/\eps^2)}\cdot n^{\Oh(1)}$.
 
 It remains to prove the second postulated condition. Let then $\Ff\subseteq \Allowed(\Pp,\Kk,f)$ be the maximum-size independent set that is inside $(\Pp,\Kk,f)$. Then $|\Ff|=\opt(\Pp,\Kk,f)$, which by assumption is at least $s$. Now \cref{lem:nice-separator} asserts that the enumeration algorithm explained above will consider a suitable choice of $\Ss$ and $C_1,C_2$ for which the constructed family $\VorCand$ will be included in $\Cand(\Pp,\Kk,f)$, and for which we have the following properties:
 \begin{itemize}[nosep]
  \item $|\Ff\cap \Allowed(\Pp_g,\Kk_g,g)|\leq \frac{3}{4}|\Ff|=\frac{3}{4}\cdot \opt(\Pp,\Kk,f)$ for all $(\Pp_g,\Kk_g,g)\in \VorCand$, and
  \item at least $(1-\eps)|\Ff|$ objects from $\Ff$ are not in conflict with $C_1$ and $C_2$.
 \end{itemize}
  Note that the first condition above in particular implies that $g\subsetneq f$ for all $(\Pp_g,\Kk_g,g)\in \VorCand$. Further, every object of $\Ff$ that is not in conflict with $C_1$ and $C_2$ is inside one of the Swiss-cheese separators $(\Pp_g,\Kk_g,g)\in \VorCand$. Hence 
 \begin{align*}
\sum_{(\Pp_g,\Kk_g,g)\in \VorCand} \min\left(\frac{3}{4}\cdot \opt(\Pp,\Kk,f),\opt(\Pp_g,\Kk_g,g)\right)&\geq \sum_{(\Pp_g,\Kk_g,g)\in \VorCand} |\Ff\cap \Allowed(\Pp_g,\Kk_g,g)|\\
& \geq (1-\eps)\cdot \opt(\Pp,\Kk,f).
 \end{align*}
 This concludes the proof.
\end{proof}

\paragraph*{Description of the algorithm.}
Now, we present our algorithm. The idea is to perform dynamic programming on
Swiss-cheese separators. First, we enumerate the set $\Sep$
using~\cref{obs:enum}. We construct two dynamic programming tables $\SmallSol,
\ApproxIS \colon \Sep \to \nat$. We initially set for every $(\Pp,\Kk,f) \in \Sep$:
\begin{displaymath}
    \SmallSol[\Pp,\Kk,f]\coloneqq
    0\qquad\textrm{and}\qquad \ApproxIS[\Pp,\Kk,f]\coloneqq 0.
\end{displaymath}
The intention is that after filling in the tables, for all
$(\Pp,\Kk,f)\in \Sep$ we will have the following:
\begin{align}
    \SmallSol[\Pp,\Kk,f]&=\min(s,\opt(\Pp,\Kk,f))\qquad\textrm{and}\nonumber\\
    \opt(\Pp,\Kk,f)^{1-4\eps}\,&\leq \ApproxIS[\Pp,\Kk,f]\leq
    \opt(\Pp,\Kk,f).\label{eq:wydra}
\end{align}
That is, in $\SmallSol[\Pp,\Kk,f]$ we store the maximum size of a solution, but
capped at the value $s$, and in $\ApproxIS[\Pp,\Kk,f]$, we store an approximation of
the maximum size of a solution. Note that the above bound will naturally lead to an
$\text{OPT}^{4\eps}$-approximation algorithm. To achieve the statement of
\cref{thm:main}, we rescale $\eps$ by a factor of $4$.

In the first phase, we compute the table $\SmallSol[\cdot]$. To do so, we
iterate over every possible separator $(\Pp,\Kk,f) \in \Sep$ and every set $\Ff \subseteq
\Allowed(\Pp,\Kk,f)$ of size at most $s$. In $\SmallSol[\Pp,\Kk,f]$, we simply store the
maximum cardinality among the considered sets $\Ff$ that turned out to be
independent. This operation can be implemented in time $|\Dd|^{\Ot(s)} \cdot
n^{\Oh(1)}$.

In the second phase, we compute the table $\ApproxIS[\cdot]$. To fill in
$\ApproxIS[\cdot]$, we iterate over separators $(\Pp,\Kk,f) \in \Sep$ according to
the partial order $\preceq$. In other words, if $f_1 \subsetneq f_2$, then we
consider the separator $(\Pp,\Kk_1,f_1)$ before the separator $(\Pp,\Kk_2,f_2)$.

Now, assuming that we are considering $(\Pp,\Kk,f) \in \Sep$, we enumerate the family $\Cand(\Pp,\Kk,f)$ provided by \cref{lem:enumerate-splits}. Then we set
$$\ApproxIS[\Pp,\Kk,f] \coloneqq \max\left(\SmallSol[\Pp,\Kk,f],\ \max_{\VorCand\in \Cand(\Pp,\Kk,f)} \left\{\sum_{(\Pp',\Kk',f')\in \VorCand}\ApproxIS[\Pp',\Kk',f']\right\} \right).$$
In other words, $\ApproxIS[\Pp,\Kk,f]$ is set to be the largest value among the candidates: $\SmallSol[\Pp,\Kk,f]$ and $\sum_{(\Pp',\Kk',f')\in \VorCand}\ApproxIS[\Pp',\Kk',f']$, for all $\Cc\in \Cand(\Pp,\Kk,f)$.

Finally, observing that $(\emptyset,\emptyset,\mathbb{S}^2)$ is also a valid
ripe Swiss-cheese separator that belongs to $\Sep$, we return
$\ApproxIS[\emptyset,\emptyset,\mathbb{S}^2]$ as the final answer. This
concludes the description of the algorithm. Note that the number of states in
our dynamic programming tables, as well as the number of operations required to
compute them, is $|\Dd|^{\Ot(1/\eps^2)} \cdot n^{\Oh(1)}$. Therefore, to prove
\cref{thm:main}, we need to justify the approximate correctness of the
algorithm. That is, we need to prove the assertion~\eqref{eq:wydra}.

\paragraph*{Approximation factor.}
We will need the following inequality.

\begin{proposition}\label{prop:inequality}
    Let $\delta, c \in (0,1)$ and $a_1,\ldots,a_m,A \in \real$
    be positive  reals such that
    $c A \ge a_i$ for all $i \in \{1,\ldots,m\}$ and $\sum_{i=1}^m a_i \ge
    c^\delta A$. Then 
    $$\sum_{i=1}^m (a_i)^{1-\delta} \ge A^{1-\delta}.$$
\end{proposition}
\begin{proof}
    By the assumption $c A \ge a_i$ we have:
    $$\LHS \coloneqq \sum_{i=1}^m (a_i)^{1-\delta} = \sum_{i=1}^m
        \frac{a_i}{(a_i)^\delta} \ge c^{-\delta} \cdot \sum_{i=1}^m
        \frac{a_i}{A^\delta}.$$
    Next, we use the assumption $\sum_{i=1}^m a_i \ge c^\delta A$ and conclude that
    $\LHS \ge c^{-\delta} \cdot \frac{c^\delta A}{A^\delta} =
    A^{1-\delta}$.
\end{proof}
Next, we proceed to the proof
of~\eqref{eq:wydra}. Note that $\SmallSol[\Pp,\Kk,f]=\min(s,\opt(\Kk,f))$ for all
$(\Pp,\Kk,f)\in \Sep$ is clear from the way we computed the entries of $\SmallSol[\cdot]$. So we are left with proving the following.

\begin{lemma}
    For every Swiss-cheese separator $(\Pp,\Kk,f) \in \Sep$ of complexity at
    most $s$, we have $$ \opt(\Pp,\Kk,f)^{1-4\eps}\leq
    \ApproxIS[\Pp,\Kk,f] \leq \opt(\Pp,\Kk,f).$$
\end{lemma}

\begin{proof}
    For the upper bound, a straightforward inductive argument over separators of
    $\Sep$ ordered by $\preceq$ shows that whenever $\ApproxIS[\Pp,\Kk,f]$ is
    set to some value, say $w$, then this is witnessed by the existence of an
    independent set of size $w$ that is inside $(\Pp,\Kk,f)$. It
    follows that $\ApproxIS[\Pp,\Kk,f] \leq \opt(\Pp,\Kk,f)$ for all $(\Pp,\Kk,f)\in \Sep$.

    We proceed to the proof of the lower bound, which is also by induction on
    $(\Pp,\Kk,f)\in \Sep$ according to the order
    $\preceq$. Note that whenever
    $\opt(\Pp,\Kk, f) \le s$, we have $\ApproxIS[\Pp,\Kk,f] = \opt(\Pp,\Kk,f)$ due to taking $\SmallSol[\Pp,\Kk,f]$ into account among the candidates for the value of $\ApproxIS[\Pp,\Kk,f]$. So assume otherwise, that $\opt(\Pp,\Kk, f)>s$.
    In that case, the second point of Lemma~\ref{lem:enumerate-splits} asserts that there exists $\VorCand\in \Cand(\Pp,\Kk,f)$ such that 
    \begin{equation}\label{eq:nutria}
    \sum_{(\Pp',\Kk',f')\in \VorCand} \min\left(\frac{3}{4}\cdot \opt(\Pp,\Kk,f),\opt(\Pp',\Kk',f')\right)\geq (1-\eps)\cdot \opt(\Pp,\Kk,f).                                                                                                                                                                                                                                                                          \end{equation}
    As $f'\subsetneq f$ for all $(\Pp',\Kk',f')\in \VorCand$, by induction we know that
    \begin{equation}\label{eq:manat}\opt(\Pp',\Kk',f')^{1-4\eps}\leq \ApproxIS[\Pp',\Kk',f']\qquad\textrm{for all }(\Pp',\Kk',f')\in \VorCand.
    \end{equation}
    We may now use \cref{prop:inequality} with 
    $$\delta \coloneqq 4\eps,\quad c \coloneqq
    \frac{3}{4},\quad A \coloneqq
    \opt(\Pp,\Kk,f),\quad \textrm{and}\quad a_i\coloneqq
    \min\left(\frac{3}{4}\cdot \opt(\Pp,\Ff,f),\opt(\Pp_i,\Ff_i,f_i)\right),$$
    where $\{(\Pp_i,\ff_i,f_i)\colon i=1,\ldots,m\}$ is an arbitrary enumeration
    of $\VorCand$.
    Indeed, for such a choice we have on one hand $cA\geq a_i$ for all $i\in \{1,\ldots,m\}$, and on the other hand, by Bernoulli's inequality\footnote{We use the following variant: $(1-x)^r\leq 1-rx$ for all $r,x\in [0,1]$.}, 
    $$1-\eps\geq \left(1-\frac{1}{4}\right)^{4\eps} =c^\delta,$$
    hence~\eqref{eq:nutria} implies that $\sum_{i=1}^m a_i\geq c^\delta A$. By \cref{prop:inequality} we conclude that
    \begin{equation}\label{eq:mors}\sum_{i=1}^m \opt(\Pp_i,\Kk_i,f_i)^{1-4\eps} \geq \sum_{i=1}^m (a_i)^{1-4\eps} \geq A^{1-4\eps}=\opt(\Pp,\Ff,f)^{1-4\eps}.\end{equation}
    As $\Cc\in \Cand(\Pp,\Ff,f)$, the algorithm considers $\Cc$ when computing $\ApproxIS[\Pp,\Ff,f]$, and hence
    \begin{equation}\label{eq:foka}
    \ApproxIS[\Pp,\Ff,f]\geq \sum_{i=1}^m \ApproxIS[\Pp_i,\Kk_i,f_i]
    \end{equation}
    We may now combine \eqref{eq:manat},~\eqref{eq:mors}, and~\eqref{eq:foka} to conclude that $\ApproxIS[\Pp,\Ff,f]\geq \opt(\Pp,\Ff,f)^{1-4\eps}$, as required.
   So the induction proof is complete.
\end{proof}

This concludes the proof of~\cref{thm:main}.

%% file: chapters/appendix.tex
\section{Proofs from~\cref{sec:separator}}\label{sec:appendix-a}

This section contains the missing proofs from~\cref{sec:separator}. We emphasize
again that to a large extent, the proofs of \cref{lem:sample-size,lem:no-heavy-spoke,lem:no-heavy-diamond} below follow~\cite[Section
2.2]{esa18}.

\begin{proof}[Proof of~\cref{lem:no-heavy-spoke}]
    {\bf{The remainder of this proof is in its entirety a paraphrase and an adaptation of the proof of \cite[Claim~2 in the proof of Lemma~4]{esa18}, with a majority of text taken verbatim.}}
    
    By~\cref{lem:types-are-enough} there are three cases to be analysed based
    on whether the spoke in question comes from a face of type 1, type 2, or type 3.
    For the sake of presentation, let us first focus on the case of type 1. 

    Consider any triple $p_1,p_2,p_3$ of different objects from $\Ff \cup \Pp$ and let $f$ be a singular type-1 face  for
    $(p_1,p_2,p_3)$. Suppose further that $P$ is one of the spokes incident to
    $f$ induced in the Voronoi diagram of $\{p_1,p_2,p_3\}$.
    Importantly, let us assume that spoke $P$ is heavy.
    Let us estimate the probability of the following event

    \begin{displaymath}
        \mathrm{Event\; \textbf{A}}: p_1,p_2,p_3 \text{ all belong to } \Ss \text{ and } P \text{ is a spoke induced in } \Diag_{\Ss}.
    \end{displaymath}

    For the Event \textbf{A} to happen two things must happen: (i) $p_1,p_2,p_3 \in \Ss$,
    and (ii) for each $p' \in \Ff \setminus \{p_1,p_2,p_3\}$ that is in conflict
    with $P$, the object $p'$ must not be included in $\Ss$.
    Let $\mathcal{Z}$ be the set of such objects $p'$. Observe that the
    probability that object $p'$ is not included in $\Ss$ is $(1-\lambda)$.
    As for (i), let $q_{p}$ be the probability that $p \in \Ss$. Observe that
    $q_p = \lambda$ if $p \in \Ff$ and $q_p=1$ if $p\in \Pp$.
    Using the standard inequality $1-x \le \exp(-x)$ we have:
    \begin{align}
        \Prob(\textbf{A}) \le  & q_{p_1}q_{p_2}q_{p_3} \cdot \prod_{p' \in
        \mathcal{Z}} (1-\lambda) \le 
    q_{p_1}q_{p_2}q_{p_3} \cdot \exp\left(-\sum_{p' \in \mathcal{Z}}
    \lambda\right)\nonumber \\
        \le & q_{p_1}q_{p_2}q_{p_3} \cdot \exp\left(- \lambda|\mathcal{Z}|\right) \le
        q_{p_1}q_{p_2}q_{p_3} \cdot \exp\left(-\lambda \eta\right) =
        q_{p_1}q_{p_2}q_{p_3} \cdot \eps^{100} \label{eq:prob_A}
    \end{align}

    For a fixed triple $p_1,p_2,p_3 \in \Ff \cup \Pp$, face $f$ that is singular
    of type-1 for $(p_1,p_2,p_3)$ and a spoke $P$ incident to $f$, let us denote
    by $\textbf{A}^1_{P,(p_1,p_2,p_3)}$ the event $\textbf{A}$ considered above.
    Let $\textbf{B}^1$ be the event that some event $\textbf{A}^1_{P,(p_1,p_2,p_3)}$ happens; that is, $\textbf{B}^1$ is the union of events $\textbf{A}^1_{P,(p_1,p_2,p_3)}$ as above.
    By~\cref{lem:few-singular-faces}, for every triple $p_1,p_2,p_3 \in \Ff \cup \Pp$ there are at most two faces $f$ that are singular of type 1 for
    $(p_1,p_2,p_3)$, and for each of them there are at most $3$ spokes $P$ to
    consider. 

    Thus, by applying the union bound to~\eqref{eq:prob_A} we get that the
    probability that some event $\textbf{A}^1_{P,(p_1,p_2,p_3)}$ happens can be bounded as follows:
    \begin{align*}
        \Prob(\textbf{B}^1) & \le 6 \eps^{100} \sum_{p_1 \in \Ff \cup \Pp}\ \sum_{p_2
        \in \Ff \cup \Pp}\ \sum_{p_3 \in \Ff \cup \Pp} q_{p_1}q_{p_2}q_{p_3}  \\
        & =6 \eps^{100} \cdot \left(\sum_{p_1 \in \Ff \cup \Pp} q_{p_1}\right) \left(\sum_{p_2 \in
        \Ff \cup \Pp} q_{p_2}\right) \left(\sum_{p_3 \in \Ff \cup \Pp} q_{p_3}\right) \\
        & \le 6 \eps^{100} \cdot \left(\sum_{p \in \Ff \cup \Pp} q_{p}\right)^3 = 6
        \eps^{100} \cdot \left(\left(\sum_{p \in \Ff} q_p\right) + \left(\sum_{p
                \in \Pp}
        q_p\right)\right)^3 = 6 \eps^{100} \left(s+|\Pp|\right)^3 
    \end{align*}
    As $|\Pp| \le 5 s$, we have
    $\Prob(\textbf{B}^1) \le 5 \cdot 48 \eps^{100} s^3 \leq 240 \eps$, which bounded by $10^{-3}$ due to $\eps<10^{-100}$. Note that if $\textbf{B}^1$ does not occur, then $\Diag_\Ss$ has no heavy spoke that would be incident to a type-$1$ singular face for some triple of objects in $\Ss$.

    This concludes the case of spokes incident to singular faces of type $1$.  Next, we analyse the remaining cases.
    Denote by $\textbf{B}^2$ the following event: there exist distinct
    $p_1,p_2,p_3 \in \Ff\cup \Pp$ and a type-$2$ singular face $f$ for $(p_1,p_2,p_3)$
    such that (i) $p_1,p_2,p_3 \in \Ss$, and (ii) at least one of the heavy
    (w.r.t. $\Ff$) spokes incident to $f$ in the Voronoi diagram of
    $\{p_1,p_2,p_3\}$ remains a spoke in $\Ss$. Similar calculation as for
    $\textbf{B}^1$ yield that $\Prob(\textbf{B}^2) \le 10^{-3}$, as the only
    property of the face of type 1 that we have used is that it was defined by a triple of objects.

    Finally, denote by $\textbf{B}^3$ the following event: there exist distinct
    $p_0,p_1,p_2,p_3 \in \Ff\cup \Pp$ and a singular type-3 face $f$ for
    $(p_0,p_1,p_2,p_3)$ such that (i) $p_0,p_1,p_2,p_3 \in \Ss$ and (ii) at
    least one of the heavy (w.r.t. $\Ff$) spokes incident to $f$ in the Voronoi
    diagram of $\{p_0,p_1,p_2,p_3\}$ remains a spoke in $\Ss$. To bound the
    probability of $\textbf{B}^3$ we apply similar calculations, however we
    sum over quadruples of objects instead of triples. This~gives:
    \begin{align*}
        \Prob(\textbf{B}^1) & \le 8 \eps^{100} \sum_{p_0 \in \Ff \cup \Pp}
        \sum_{p_1 \in \Ff \cup \Pp} \sum_{p_2
        \in \Ff \cup \Pp} \sum_{p_3 \in \Ff \cup \Pp} q_{p_0}q_{p_1}q_{p_2}q_{p_3}  \\
        &  \le  8\eps^{100} \cdot \left(\left(\sum_{p \in \Ff} q_p\right) +
            \left(\sum_{p \in \Pp}
        q_p\right)\right)^4 \le 8 \eps^{100} \left(s+|\Pp|\right)^4.
    \end{align*}
    Again, as $|\Pp| \le 5s$, this means that $\Prob(\textbf{B}^3) \le 5 \cdot 128 \eps
    \le 10^{-3}$. By~\cref{lem:types-are-enough} we conclude that the
    probability that any spoke in the diagram $\Diag_{\Ss}$ is heavy is at most $3 \cdot 10^{-3}
    \le \frac{1}{4}$.
\end{proof}

\begin{proof}[Proof of~\cref{lem:no-heavy-diamond}]
    {\bf{The remainder of this proof is in its entirety a paraphrase and an adaptation of the proof of \cite[Claim~3 in the proof of Lemma~4]{esa18}, with a majority of text taken verbatim.}}

    Consider an edge $e$ of $\Diag_\Ss$. Let $f_1,f_2$ be the endpoints of $e$ and let
    $p_1,p_2$ be the objects of $\Ss$ so that $\diamondsuit_\Ss(e)$ corresponds ot the $4$-cycle $p_1f_1p_2f_2$ in $\Rad_\Ss$. From~\cref{lem:types-are-enough} we have that both $f_1$ and $f_2$ are singular faces, for a triple (for types 1 or 2) or a quadruple (for type 3) of objects in $\Ss$. Consider first the case when $f_1$ and $f_2$ are type-1 singular faces for triples of objects from
    $\Ss$, then $f_1$ is a type-1 singular face for 
    $(p_1,p_2,p_1')$ for some object $p_1' \in \Ss$ and $f_2$ is a type-1
    singular face for $(p_1,p_2,p_2')$ for some object $p_2' \in \Ss$. Let us assume for a moment that $p_1' \neq p_2'$; all the other cases will be discussed later, and they will follow the same reasoning scheme.

    We have a quadruple of pairwise different objects $p_1,p_2,p_1',p_2'$. For
    the diamond $\diamondsuit = \diamondsuit_\Ss(e)$ to arise in the diagram $\Diag_\Ss$, both of the
    following two events must happen: (i) objects $p_1,p_2,p_1',p_2'$ are
    included in $\Ss$, and (ii) all objects intersecting
    $\diamondsuit$ except for $p_1,p_2$ must not be included in $\Ss$.
    These two events are independent. The probability of the first is $q_{p_1}
    \cdot q_{p_2} \cdot q_{p_1'} \cdot q_{p_1'}$, where $q_p = \lambda$ if $p
    \in \Ff$ and $q_p = 1$ if $p \in \Pp$. Therefore, we are left to analyse
    the probability of event (ii).

    Let $\Gg$ be the family of objects entirely contained in the interior of
    $\diamondsuit$. By the assumption that $\diamondsuit$ is heavy, we know that $|\Gg| \ge
    \eta$. Therefore, the probability that no object of $\Gg$ is included in
    $\Ss$ is upper bounded by
    \begin{align*}
        \prod_{r \in \Gg} (1-q_r) \le \exp \left( -\sum_{r \in \Gg} q_r \right)
        \le \exp(-\eta\lambda) =  \eps^{100}.
    \end{align*}

    We conclude that the probability that $\diamondsuit$ arises in $\Diag_\Ss$ is upper
    bounded by $q_{p_1} q_{p_2} q_{p_1'} q_{p_1'} \cdot \eps^{100}$.
    For every quadruple $(p_1,p_2,p_1',p_2')$ of objects in $\Ff \cup \Pp$, there
    are at most $4$ diamonds induced by this quadruple as above, as there are at
    most $2$ type-1 singular faces for $(p_1,p_2,p_1')$, and similarly, for
    $(p_1,p_2,p_1')$. Therefore, the total probability that in $\Diag_\Ss$ there is a
    heavy diamond with both endpoints being singular faces of type 1 and $p_1'
    \neq p_2'$ is bounded by:
    \begin{align*}
        4 \eps^{100} \cdot \sum_{p_1,p_2,p_1',p_2' \in \Ff \cup \Pp} q_{p_1}
        q_{p_2} q_{p_1'} q_{p_2'} & \le 4\eps^{100} \cdot \left(\sum_{p \in
        \Ff \cup \Pp} q_p \right)^4\\
        & \le 4 \eps^{100} \cdot \left(|\Pp| + s\right)^4 \le 4 \cdot \eps \le 10^{-3}.
    \end{align*}

    The reasoning for the case when $p_1' = p_2'$ is analogous; the bound is
    even better as we sum over triples instead of quadruples. The proofs for the cases when $f_1$ or $f_2$ is a face of type $2$ or $3$ are similar; the type $3$ we would consider a quadruple of objects for the respective branching point instead of a triple, so for instance the case when both $f_1$ and $f_2$ are of type $3$ would involve summation over $8$-tuples of objects. Due to the appearance of the 
    $\eps^{100}$ factor, the probability in every case can be bounded by $10^{-3}$. Moreover, the number of different cases is bounded by $250$, so summing all the
    probabilities we can conclude that the probability that $\Diag_\Ss$ contains a
    heavy diamond is at most $\frac{1}{4}$.
\end{proof}

%% file: chapters/appendix-sep.tex
\section{Finding balanced cycles in the radial graph} \label{sec:proof-separator}

This section is dedicated to the proof of the following separator statement for Voronoi diagrams.

\separator*

Our proof of \cref{thm:separator} is a rather straightforward modification of the proof of~\cite[Theorem~4.15 and Lemma~4.18]{esa15}.
The main caveat is that in \cite[Section~4]{esa15}, Pilipczuk and Marx consider only the setting of unweighted graphs, and need balancedness with respect to either the number of edges or of faces. Here, we consider the vertex-weighted setting, hence the reasoning has to be adjusted. 
The remainder of this section is dedicated to the proof of \cref{thm:separator}.

\paragraph*{Nooses.} Before we start, we need to recall the notion of a \emph{noose} and its main properties; {\bf{this recollection is a close-to-verbatim citation of the discussion in~\cite[Section~4.6]{esa15}}}. 
Let $H$ be a connected graph embedded in the sphere $\mathbb{S}^2$. A {\em{noose}} in $H$
is a closed, directed curve $\sigma$ on $\mathbb{S}^2$ without self-crossings
that only intersects $H$ at its vertices and visits every face of $H$ at most
once. When we remove $\sigma$ from the sphere $\mathbb{S}^2$, it divides the
sphere into two open disks: one for which $\sigma$ represents the clockwise
traversal of the perimeter, and the other for which it represents the
counterclockwise traversal (assuming a fixed orientation of $\mathbb{S}^2$).
Following the notation of \cite{esa18}, we shall call the first disk
$\enc(\sigma)$ (for \emph{enclosed}), and the second disk we call $\exc(\sigma)$
(for \emph{excluded}). The {\em{length}} of a noose is the number of vertices of $H$ it~crosses.

We will use the following variant of a balanced noose lemma.

\begin{lemma}\label{lem:noose-separator}
    Let $H$ be a connected $3$-regular multigraph embedded in the sphere $\mathbb{S}^2$. Let $\omega$ be an assignment of nonnegative weights to
    the vertices of $H$ such that the total sum of the weights is $1$. Then there exists a noose
    $\sigma$ of length at most $\sqrt{\frac{9}{2}|V(H)|}$ such that the total weight of vertices lying inside
    $\enc(\sigma)$ is at most $2/3$, and the same holds for the vertices lying inside $\exc(\sigma)$.
\end{lemma}

The proof of \cref{lem:noose-separator} is essentially a repetition of the proof of
\cite[Theorem~4.15]{esa15}, so we only highlight the necessary minor modifications. There are two differences between \cref{lem:noose-separator} and \cite[Theorem~4.15]{esa15}:
\begin{enumerate}[(i),nosep]
 \item in \cite[Theorem~4.15]{esa15} the balancedness is measured in terms of edges of $H$ instead of vertices; and
 \item \cite[Theorem~4.15]{esa15} considers the unweighted setting where every edge of $H$ has a unit weight.
\end{enumerate}
These two differences can be mitigated as follows. 

In the proof of \cite[Theorem~4.15]{esa15} one considers the decomposition of
$H$ into bridgeless components and identifies a {\em{central}} bridgeless
component $B_{x_0}$: a bridgeless component such that the total weight of every
component of $G-B_{x_0}$ is at most $1/2$. There is an easy corner case when
$B_{x_0}$ consists of just one vertex, which can be resolved in exactly the same
way in our setting when the diagram is vertex-weighted. So let us assume that
$B_{x_0}$ consists of more than one vertex. 

Observe that if for any component $C$ of $G-B_{x_0}$ the total weight of vertices in $C$ is between $1/3$ and $1/2$, then we can easily construct a noose of length $1$ with the desired balancedness property as follows: the noose travels through the unique face lying on both sides of the bridge connecting $C$ and $B_{x_0}$, and crosses the endpoint of that bridge that belongs to $B_{x_0}$. Hence, we may assume that for each component $C$ of $G-B_{x_0}$, the total weight of the vertices of $C$ is at most $1/3$.

Next, in \cite[Theorem~4.15]{esa15} one considers the induced subgraph $G[B_{x_0}]$ (we consider $B_{x_0}$ to be a subset of vertices of $G$) and constructs a weight function on the edges of this subgraph as follows. For every connected component $C$ of $G-B_{x_0}$, say connected to $B_{x_0}$ through a bridge $fg$ where $g\in B_{x_0}$, we assign all the edges of $C$ (including the bridge $fg$) to one of the edges of $G[B_{x_0}]$ incident to $g$. Then the weight of an edge $e\in E(G[B_{x_0}])$ is equal to $1$ plus the total number of edges assigned to $e$. In our vertex-weighted setting, we can just assign the total weight of vertices in $C$ to the vertex $g$, so that its new weight becomes the original weight plus the total weight of vertices in $C$.

Now that $G[B_{x_0}]$ is bridgeless, one can find a balanced noose (which is a noose in the sphere drawing of $G[B_{x_0}]$) using known results about the existence of {\em{sphere-cut decompositions}} in bridgeless planar multigraphs. In~\cite[Theorem~4.15]{esa15} the noose is balanced with respect to the edge weights, but exactly the same proof strategy also yields a noose $\gamma$ such that the total weight of the vertices of $B_{x_0}$ strictly enclosed by $\gamma$ is at most $2/3$, and the same holds also for the vertices of $B_{x_0}$ strictly excluded by $\gamma$.

The last step in the proof is to lift the noose $\gamma$ --- which is a noose in the sphere drawing of $G[B_{x_0}]$ --- to a noose $\gamma'$ in the sphere drawing of $G$. This boils down to deciding, for every component $C$ of $G-B_{x_0}$ such that $\gamma$ crosses the vertex of $B_{x_0}$ adjacent to $C$, whether $C$ should be placed so that it is entirely enclosed or entirely excluded by $\gamma'$. Since the total weight of every such $C$ is at most $1/3$, it is not difficult to see that by investigating the components $C$ one by one, and always assigning them to the side of $\gamma'$ that is currently lighter, we do not break the $2/3$-balancedness property. Hence, the noose $\gamma'$ constructed at the end will be $2/3$-balanced, as required.

\medskip

Equipped with \cref{lem:noose-separator}, we continue the proof of \cref{thm:separator}. Let $\sigma$ be the noose given by \cref{lem:noose-separator} for the diagram $H\coloneqq \Diag_\Ss$ (which, recall, is a connected sphere-embedded $3$-regular multigraph).
Analogously to \cite{esa15}, we define a cycle $C=C(\sigma)$ in $\Rad_\Ss$ as follows. Suppose $\sigma$ visits faces and branching
points $f_1^\ast, f_1, f_2^\ast, f_2, \ldots, f_r^\ast, f_r$ in this order,
where $f_i^\ast$ are different faces of $H$, and $f_i$ are different branching
points of $H$. Recall that faces of $H$ correspond to objects of $\Ff$. Therefore,
let $p_1, \ldots, p_r \in \Ff$ be the objects that correspond to $f_1^\ast,
\ldots, f_r^\ast$, respectively.

Next, consider any $i \in \{1, \ldots, r\}$. Upon leaving face $f_i^\ast$, noose $\sigma$ enters branching point $f_i$ through one of three regions in a close neighborhood of $f_i$, separated by the edges of $H$ incident to $f_i$. Each of these three regions corresponds in a natural way to one of the vertices of $f_i$ (which, recall, is a triangular face of $G$); see also the discussion in
\cite[Section 4.8]{esa15} for a formal treatment. Therefore, we define $u_i$ to be the vertex of $f_i$ that corresponds to the region through which $\sigma$ enters $f_i$ from $f_i^\ast$. Analogously, we define $v_i$ to be the vertex of $f_i$ that corresponds to the region through which $\sigma$ enters face $f_{i+1}^\ast$ from branching point $f_i$ (with indices behaving cyclically). 


We now may define the cycle $C=C(\sigma)$ naturally: the consecutive edges of $\Diag_\Ss$ traversed by $C$ are the edge $p_if_i$ with label $u_1$ and the edge $f_ip_{i+1}$ with label $v_1$, for $i=1,\ldots,r$ (indices behave cyclically). Since $\sigma$ visits every face of $\Diag_\Ss$ and every branching point of $\Diag_\Ss$ at most once, it is clear that $C$ is a cycle in $\Diag_\Ss$. 
Moreover, the balancedness assertion provided by \cref{lem:noose-separator} and the construction immediately imply the balancedness property asserted in the statement of \cref{thm:separator}. Finally, the length of $C$ is twice larger than the length of $\sigma$, hence it is bounded by $\sqrt{18|V(H)|}$. So the proof of \cref{thm:separator} is complete.